\documentclass[12pt]{article}
\usepackage{sectsty}
\usepackage[T1]{fontenc}
\usepackage{kpfonts}
\usepackage[dvipsnames]{xcolor}
\usepackage{pdfpages}
\usepackage{accents}
\usepackage{tikz-cd}
\usepackage{float}
\usepackage[round]{natbib}
\usepackage{multirow}
\usepackage{dutchcal}
\usepackage{bbm}
\usepackage{sgame}
\usepackage{mathtools}
\usepackage{amsthm}
\usepackage{enumitem}
\usepackage[margin=1in]{geometry}
\usepackage{caption}
\usepackage{subcaption}
\usepackage{epigraph}
\usepackage{listings}
\usetikzlibrary{calc}
\usetikzlibrary{shapes,arrows}
\usepackage{tcolorbox}
\usepackage{marvosym}

\DeclareMathOperator*{\supp}{supp}
\DeclareMathOperator*{\conv}{ch}

\DeclareMathOperator*{\interior}{int}
\DeclareMathOperator*{\ch}{ch}
\DeclareMathOperator*{\aff}{aff}
\newcommand{\R}{\mathbb{R}}

\renewcommand{\P}{\mathcal{P}}
\newcommand{\e}{\varepsilon}
\newcommand{\mut}{\tilde{\mu}}
\newcommand{\muh}{\hat{\mu}}

\newcommand{\0}{\underline{0}}

\newcommand{\norm}[1]{\lVert #1\rVert}

\newtheorem{theorem}{Theorem}[section]
\newtheorem{proposition}{Proposition}
\newtheorem{lemma}{Lemma}
\newtheorem{corollary}{Corollary}

\theoremstyle{definition}

\newtheorem{definition}[theorem]{Definition}
\newtheorem{example}[theorem]{Example}

\newtheorem{remark}[theorem]{Remark}

\usepackage{hyperref}
\hypersetup{
    colorlinks=true,
    linkcolor=Violet,
    filecolor=Violet,      
    urlcolor=Violet,
    citecolor=Violet,
}

\definecolor{backcolour}{rgb}{0.63, 0.79, 0.95}

\lstdefinestyle{mystyle}{
  backgroundcolor=\color{backcolour},
  basicstyle=\ttfamily\footnotesize,
  breakatwhitespace=false,         
  breaklines=true,                 
  captionpos=b,                    
  keepspaces=true,                 
  numbers=left,                    
  numbersep=5pt,                  
  showspaces=false,                
  showstringspaces=false,
  showtabs=false,                  
  tabsize=2
}
\tikzset{
  > = latex',
  axis/.style    = {very thick},
  aborder/.style = {draw},
  acomp/.style   = {fill=black, fill opacity=0.1},
  rect/.style    = {very thick},
  form/.style    = {font=\scriptsize},
  sm/.style      = {font=\small},
  vsm/.style     = {font=\scriptsize}
}

\DeclareMathOperator{\cav}{cav}

\usepackage{thmtools,thm-restate}

\begin{document}

\author{Andreas Kleiner \and Benny Moldovanu \and Philipp Strack \and Mark
Whitmeyer\thanks{We thank participants at ESSET 2022 (Gerzensee), the 2023 SEA conference (New Orleans), and the 2024 SAET conference (Santiago) for their comments. Kun Zhang provided excellent research assistance. Kleiner, Moldovanu, and Whitmeyer acknowledge
financial support from the DFG (German Research Foundation) via Germany's
Excellence Strategy - EXC 2047/1 - 390685813, EXC 2126/1-390838866 and the
CRC TR-224 (projects B01 and B02). Strack was supported by a Sloan Fellowship.
Kleiner: Department of Economics, University of Bonn, \href{mailto:andreas.kleiner@asu.edu}%
{andykleiner@gmail.com}; Moldovanu: Department of Economics, University of
Bonn, \href{mailto:mold@uni-bonn.de}{mold@uni-bonn.de}; Strack: Department
of Economics, Yale University, New Haven, \href{mailto:philipp.strack@yale.edu}%
{philipp.strack@yale.edu}; Whitmeyer: Department of Economics, Arizona State
University, \href{mailto:mark.whitmeyer@gmail.com}{mark.whitmeyer@gmail.com}.%
}}
\title{The Extreme Points of Fusions}
\maketitle

\begin{abstract}
Our work explores fusions, the multidimensional counterparts of mean-preserving contractions and their extreme and exposed points. We reveal an elegant geometric/combinatorial structure for these objects. Of particular note is the connection between Lipschitz-exposed points (measures that are unique optimizers of Lipschitz-continuous objectives) and power diagrams, which are divisions of a space into convex polyhedral ``cells'' according to a weighted proximity criterion. These objects are frequently seen in nature--in cell structures in biological systems, crystal and plant growth patterns, and territorial division in animal habitats--and, as we show, provide the essential structure of Lipschitz-exposed fusions. We apply our results to several questions concerning categorization.
\end{abstract}

\section{Introduction}

We study the extreme points of the set of fusions of a given measure, the
multidimensional analogs of extreme mean-preserving contractions. The set
of fusions of a given measure $\mu $---taken here to be absolutely
continuous with respect to the Lebesgue measure---is the set of probability
distributions that are dominated by $\mu $ in the convex stochastic order,
and are thus "less variable" than $\mu $.

The burgeoning literature on information design and Bayesian persuasion
highlights the importance of fusions for economic applications: by a famous
result that goes back to the work of \cite{blackwell1953equivalent} and \cite{strassen1965existence},
the set of fusions of a given measure $\mu $ is also the set of
distributions of posterior means that can be induced by some "signal" (or
"experiment") starting from a prior belief $\mu $ about the state of the
world. Thus, optimization over the set of feasible signals of an objective
function that depends linearly or convexly on the posterior mean can be
reduced to an optimization over the set of fusions of $\mu $. By Bauer's
Theorem, the optimum in such an exercise---representing an optimal way to
reveal information about the state---will be achieved at an extreme point of
the set of fusions. This feature constitutes the primary motivation
for our study.

The case where the underlying state of the world is unidimensional has been
studied along the above lines by \cite{kleiner2021extreme} and \cite{arieli2023optimal}: in that case, each extreme measure is characterized by a
collection of (convex) intervals where either the mass put by the prior on
the interval remains in place (this corresponds to full information
revelation of these states in a Bayesian-persuasion problem), or where this
mass is contracted into a measure with at most two elements in its support
(this corresponds to pooling of states where information is only partially
or not-at-all revealed). Here we extend these insights to the
multidimensional case.

Our first main result, Proposition \ref{prop:necessity_extreme} in Section \ref{section2} focuses on finitely
supported, extreme fusions of a measure $\mu $ defined on a convex, compact
subset $X$ of $\mathbb{R}^{n}$. It shows that, necessarily, each such extreme fusion $\nu \ $induces
a partition of $X$ in convex sets such that the support of $\nu $ on each
element of the partition is an affinely independent set, and such that the
restriction of $\nu $ on each element of the partition is itself a fusion of
the restriction of the prior $\mu $ on that element. The analogy to the
unidimensional result sketched above becomes clear by noting that the
maximal cardinality of an affinely independent set in $\mathbb{R}^{n}$ is $n+1$ and this equals two for states on the real line.

Our second main result, Proposition \ref{prop:characterization_strongly_exposed} in Section \ref{section3}, offers a necessary and
sufficient condition for a finitely supported measure $\nu $ to be the
unique maximizer, on the set of fusions of a given measure $\mu$, of a
linear functional induced by a Lipschitz-continuous objective function. We
call such maximizers \textit{Lipschitz-exposed} since in convex analysis 
\textit{exposed }points are those extreme points that are unique maximizers
of some linear functional.

Our result characterizes finitely supported Lipschitz-exposed fusions in
terms of very particular convex partitions that form a \textit{%
power diagram. }A power diagram in $\mathbb{R}^{n}$ is generated by a set of \textit{sites} and by a set of \textit{%
weights, }one\textit{\ }attached to each site\textit{. }It is the partition
of $\mathbb{R}^n$ into (necessarily convex) subsets such that each element in the
partition represents the set of points that are closer in Euclidean distance
(modulo the linearly additive weight) to a given site than to any other site.

In particular, the somewhat better known \textit{Voronoi diagrams} are power
diagrams where all weights attached to sites are equal. In a power diagram
the boundary between two elements of a partition is necessarily
perpendicular to the line connecting the two respective sites. Modifying the
weights while allowing them to differ in value induces parallel shifts of
this boundary, whereas the perpendicular boundary is completely fixed by the
location of sites in any Voronoi diagram\footnote{%
A somewhat related appearance of power diagrams is in the theory of 
\textit{semi-discrete optimal transport }(see for example the textbook \cite{galichon2018optimal}): when a continuous measure $\mu $ is optimally transported to
another one with a finite support $\nu $, the sets of points that are
transported to each point in the support of $\nu $ form a power diagram. }.

The above emergence of elegant combinatorial/geometrical structures in
our study is rather surprising---it cannot be gleaned from the
unidimensional environment because there each convex partition is trivially a power
diagram. The proof of Proposition \ref{prop:characterization_strongly_exposed} combines insights from duality developed
by \cite{dworczak2019simple} and \cite{dworczak2019persuasion} with a fundamental result due to \cite{aurenhammer1987power} that establishes the equivalence between power diagrams
and regular polyhedral divisions.

Recall that in the unidimensional case, an extreme fusion allows mass to stay
in place on some intervals (corresponding to full revelation of
information). For the multidimensional case, Corollary \ref{cor:suff_for_exposed} (to Proposition \ref{prop:characterization_strongly_exposed})
allows mass to stay in place within some partition elements of a\ power
diagram and offers a sufficient condition for obtaining a Lipschitz-exposed
point that may not have finite support. The proof of this Corollary requires
a non-trivial result from the theory of the Monge-Ampere partial
differential equation. It is still an open question whether \textbf{all}
Lipschitz-exposed points are indeed characterized by a power diagram where
mass in each element of the partition either stays in place or is fused to
an affinely independent set of points.

In the remainder of Section \ref{section3} we look at the gap between the necessary
conditions for finitely supported extreme points and the sufficient
conditions for finitely supported Lipschitz-exposed points. In Section \ref{section4} we study \textit{convex-partitional} extreme points
where the mass in each element of a partition is fused into a single point, 
and their application to canonical solutions to moment persuasion. These
special fusions correspond to monotone-partitional signals that were
analyzed in the unidimensional case by several authors, e.g., \cite{ivanov2021optimal}. In Section \ref{section5} we apply our results to the study of categorization. Appendix \ref{proofs} contains the proofs and Appendix \ref{sec:appendix_B} contains finer sufficient conditions for extremal fusions.

\citet{dworczak2019persuasion} offer an elegant, unified duality approach to Bayesian Persuasion. In particular, they offer a sufficient condition for the optimality of convex-partitional signals. Beyond those cases, they also prove a structural result of interest here, their Theorem 8: under the assumptions that the objective is continuously differentiable and that the prior has a density with respect to the Lebesgue measure, the moment persuasion problem has a solution defined by partition of the state space into convex regions; moreover, for each element of the partition, any induced posterior mean is an extreme point of the (possibly infinite) set of posterior means induced by the optimal signal for that region. Observe that in one dimension---where any set has at most two extreme points---Dworczak and Kolotilin’s Theorem 8 reduces to the ``bipooling'' structure of extreme points of the set of mean-preserving contractions for a given prior, obtained by \citet{kleiner2021extreme} and \citet{arieli2023optimal}. We note that our Proposition 2 uses the duality approach of Dworczak and Kolotilin and reveals a finer structure of finitely supported, Lipschitz-exposed fusions: the induced convex partition of the state space has in fact a special geometric structure, namely is a power diagram. This additional structure can be useful in applications because power diagrams can be parametrized by relatively few parameters, which simplifies the optimization problem.

\subsection{Preliminaries and Notation}

We begin by introducing the notation used below, and a few important
concepts. Let $X$ be a compact and convex subset of a finite-dimensional
Euclidean space. For $A\subseteq X$, $\interior A$ denotes the (relative)
interior of $A,$ and $\ch A$ denotes the convex hull of $A$. For a measure%
\footnote{%
Throughout, we use the term measure for what sometimes is called a signed
measure.} $\mu $ on $X$, $\supp \mu $ denotes its support, and for measurable 
$A\subseteq X$, $\mu|_A$ denotes the restriction of $\mu$ to $A$. $\delta
_{x}$ denotes a Dirac measure concentrated on $x\in X.$

The \emph{barycenter} of a measure $\mu $ on $X$ is defined by 
\begin{equation*}
r_{X}\left( \mu \right) =\frac{1}{\mu \left( X\right) }\int_{X}xd\mu \left(
x\right).
\end{equation*}%
Note that the integral in the above definition is of a vector-valued function; it means that, for any linear function $V\ $on $X$ one has 
\begin{equation*}
V(r_{X}\left( \mu \right) )=\frac{1}{\mu \left( X\right) }\int_{X}V(x)d\mu
\left( x\right) .
\end{equation*}

\begin{definition}
For measures $\mu $ and $\nu $, we say that $\mu $ \emph{dominates }$\nu $%
\emph{\ in the convex order} (or that $\nu \ $is a \emph{fusion} of $\mu ),\ 
$denoted by $\mu \succeq \nu $, if $\int \phi \,\mathrm{d}\mu \geq \int \phi
\,\mathrm{d}\nu $ for all convex functions $\phi $ such that both integrals
exist. We write $\mu \succ \nu $ if $\mu $ dominates $\nu $ in the convex
order and $\mu \neq \nu .$\ We denote by $F_{\mu }=\{\nu :$ $\nu \preceq \mu
\}$ denote the set of fusions of a given measure $\mu .$
\end{definition}

\medskip

A finite collection of vectors $V=\left\{ x_{1},\dots ,x_{k}\right\} $ in $%
\mathbb{R}^{n}$ is \emph{affinely independent} if the unique solution to $%
\sum_{i=1}^{k}\lambda _{i}x_{i}=0$ and $\sum_{i=1}^{k}\lambda _{i}=0$ is $%
\lambda _{i}=0$, $i=1,\dots ,k$. Recall also the well-known equivalence: $%
x_{1},\dots ,x_{k}$ are affinely independent if and only if $%
x_{2}-x_{1},\dots ,x_{k}-x_{1}$ are linearly independent. Linear
independence implies affine independence, but not vice-versa. In particular,
the maximal number of affinely independent vectors in $\mathbb{R}^{n}$ is $%
n+1$. A finite collection of vectors $V=\left\{ x_{1},\dots ,x_{k}\right\} $
in $\mathbb{R}^{n}$ is \emph{convexly independent} if no element $x_{i}\in V$
lies in $\ch\left( V\setminus \left\{ x_{i}\right\} \right) $.

An \emph{extreme} point of a convex set $X$ is a point $x\in X$ that cannot
be represented as a convex combination of two other points in $X$.\footnote{%
Formally $x\in A$ is an extreme point of $A$ if $x=\alpha y+(1-\alpha )z,$
for $z,y\in A$ and $\alpha \in \lbrack 0,1]$ imply together that $y=x$ or $%
z=x$.} The usefulness of extreme points for optimization stems from \emph{%
Bauer's Maximum Principle}: a convex, upper-semicontinuous functional on a
non-empty, compact, and convex set $X$ of a locally convex space attains its
maximum at an extreme point of $X.$ The \emph{Krein--Milman Theorem} states
that any convex and compact set $X$ in a locally convex space is the closed,
convex hull of its extreme points. In particular, such a set has extreme
points.

An element $x$ of a convex set $X$ is \emph{exposed\ }if there exists a
linear functional that attains its maximum on $X$ uniquely at $x$.\footnote{%
Formally, $x$ is exposed if there exists a supporting hyperplane $H$ such
that $H\cap A=\{x\}.$} Every exposed point is extreme, but the converse is
not true in general.

\section{The Finitely Supported Extreme Points of \texorpdfstring{$F_\mu$}{TEXT}}\label{section2}

Our first main result offers a necessary condition in order for a finitely
supported measure $\nu$ to be an extreme point of $F_{\mu }$. A key feature
of an extremal fusion is the partition of the domain $X$ into convex sets $%
P\ $such that all the original mass $\mu |_{P}$ remains within $P$ and is
fused into $\nu |_{P}$ whose support is an affinely independent set of
points.

\begin{proposition}
\label{prop:necessity_extreme} Let $X\subseteq \mathbb{R}^{n}$ be compact
and convex, and let $\mu $ be an absolutely continuous probability measure
on $X$. 
Suppose that $\nu $ is an extreme point of $F_{\mu }$ that is finitely
supported. Then there exists a partition $\mathcal{P}$ of $X$ into convex
sets such that, for each $P\in \mathcal{P}$, $\nu |_{P}\preceq \mu |_{P}$
and $\nu |_{P}$ has affinely-independent support.
\end{proposition}
In the finite-state Bayesian persuasion problem introduced by \cite{kamenicagentzkow2011}, there is an unknown state of the world \(\theta \in \Theta\) (with $|\Theta| = n \in \mathbb{N}$) distributed according to some full-support prior \(\mu_0 \in \Delta\left(\Theta\right)\), about which a principal can commit to an information structure. The principal's preferences over the agent's action induce a reduced-form value function \(w \colon \Delta\left(\Theta\right) \to \mathbb{R}\), and the principal's problem reduces formally to \(\max_{F} \int_{\Delta\left(\Theta\right)} w(\mu)dF(\mu)\), subject to the Bayes-plausibility (martingality) constraint \( \int_{\Delta\left(\Theta\right)} \mu dF(\mu) = \mu_0\). 

\cite{kamenicagentzkow2011} use the Fenchel-Bunt extension of Carath\'{e}odory's Theorem (\cite{fenchel1929}, \cite{bunt1934}) to show that (under a regularity specification on the objective) there exists an optimal solution to the principal's problem in which the optimal distribution has support on at most \(n\) points. An alternative way of obtaining this finding is by observing that \cite{winkler1988}'s main result can be used to argue that the extreme points of distributions supported on the \(\left(n-1\right)\)-simplex with a specified barycenter have affinely independent support. This observation provides intuition for the structure of the necessary condition given by this Proposition. If a measure is an extreme point of the set of fusions there is always a decomposition that makes it akin to a collection of probability measures with specified barycenters for which locally the support is affinely independent---this suggests a local-global decomposition of extremal fusions that we analyze in Section \ref{sec:moment_persuasion}.

\section{Exposed Fusions}\label{section3}

In this Section we prove our main result, a characterization of those
finitely supported fusions of a given distribution that are unique
maximizers of linear functionals. 

\begin{definition}
A measure $\nu \in F_{\mu }$ is a\emph{\ Lipschitz}-\emph{exposed }point%
\emph{\ }of $F_{\mu }$ if there exists a Lipschitz-continuous function $%
u\colon X\rightarrow \mathbb{R}$ such that $\nu $ is the unique solution to
the problem 
\begin{equation*}
\max_{\lambda \in F_{\mu }}\int u(x) \,\mathrm{d}\lambda (x)
\end{equation*}
\end{definition}

In order to characterize the Lipschitz-exposed points of the set of fusions
we need the following concept:

\subsection{Power Diagrams}

\begin{figure}
    \centering
    \includegraphics[width=1\linewidth]{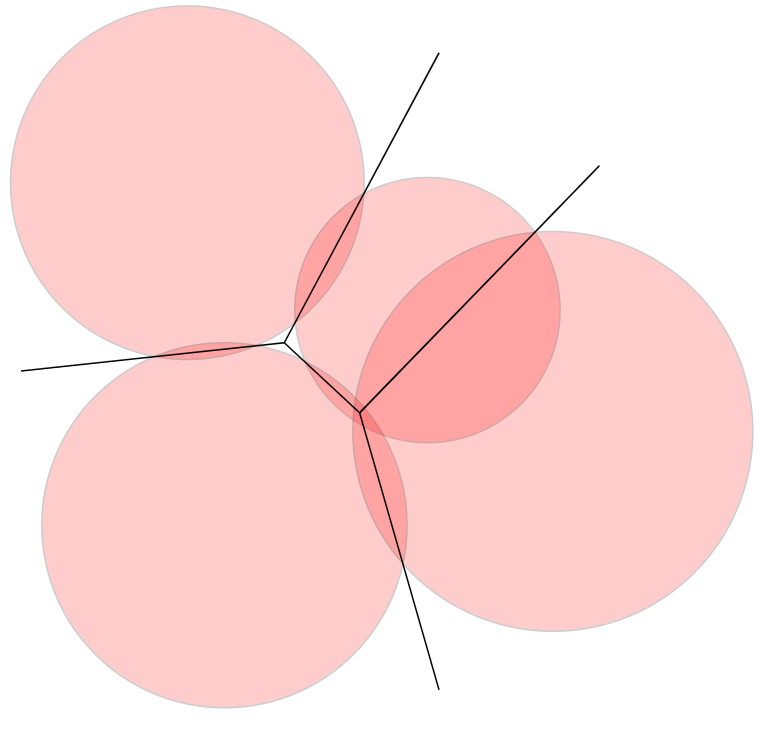}
    \caption{A power diagram (\cite{eppstein2014})}
    \label{fig:power}
\end{figure}

\begin{definition}
Consider a finite collection of \textit{sites} $S=\left\{ s_{1},\dots
,s_{k}\right\} $ in $\mathbb{R}^{n}$ and a finite collection of \textit{%
weights} $W=\left\{ w_{1},\dots ,w_{k}\right\} $ in $\mathbb{R}_{+}.$ For
each site $s_{i}$ and weight $w_{i}$ define the shifted distance function 
\begin{equation*}
g_{i}(x)=\parallel x-s_{i}\parallel^2 -\ w_{i}
\end{equation*}%
where $\parallel \cdot \parallel $ denotes the Euclidean norm. For each site 
$s_{i}$ define the \textit{cells} $P_{i}^{S,W}(x)=P_{i}(x)\subseteq $ $%
\mathbb{R}^{n}$ as following: 
\begin{equation*}
P_{i}^{S,W}(x)=\{x:\ \forall j,\ g_{i}(x)\leq g_{j}(x)\}
\end{equation*}%
The collection of cells $\mathcal{P}=\{P_{1},P_{2},..,P_{k}\}$ is the \emph{%
power diagram} generated by $(S,W)\footnote{%
Voronoi diagrams are those power diagrams where all weights are equal, e.g. $%
w_{i}=0$ for all $i.$ In contrast to Voronoi diagrams, in a power diagram a
cell may be empty, and a site may not belong to its cell.}.$
\end{definition}

For each $i,j,i\neq j,$ the set $P_{i}\cap P_{j}$ in a power diagram
is contained in a hyperplane that is perpendicular to the line connecting the sites $s_{i}$ and $s_{j}$. It is clear then that each cell of a power diagram is a polyhedron. Moreover, the intersection of any two cells in a power diagram is a common face. Every power diagram therefore corresponds to a polyhedral subdivision:
\begin{definition}
A \emph{polyhedral subdivision }$T$ of a compact and convex set $X\subseteq \R^n$ is a finite collection of polyhedra $%
K_{i} $ such that:

\begin{itemize}
\item $\cup K_{i}=X$.

\item For any $K_{i}$ in $T$, all faces of $K_{i}$ are in $T.$

\item The intersection of any two polyhedra $K_{i}$ and $K_{j}$ in $T$ is a
face of both.

\end{itemize}
\end{definition}

A polyhedral subdivision $T$ of $X$ in $\mathbb{R}^{n}\ $is \emph{%
regular} if, for each point $x\in X$ there is a \emph{height} $\alpha
_{x}\geq 0$ such that $T$ is isomorphic to the set of lower faces\footnote{%
Lower faces are in directions with a negative $(d+1)-$th coordinate.} of the
polyhedron $ch\{(x,\alpha _{x})\in \mathbb{R}^{n+1},x\in X\}.$ Equivalently, for any regular polyhedral subdivision $T$ there is a convex function that is affine on each $K_i\in T$ and, if the function is affine on a set $B$, then $B\subseteq K_i$ for some $K_i\in T$. See Figure \ref{fig:mother} for an example of a polyhedral subdivision that is not regular.

\begin{figure}
    \centering
    \includegraphics[width=0.5\linewidth]{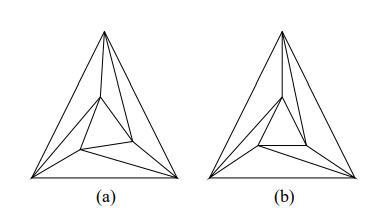}
    \caption{A regular (a) and irregular (b) subdivision (\cite{lee_santos_subdivisions_2017}).}
    \label{fig:mother}
\end{figure}

A fundamental result due to \cite{aurenhammer1987power} establishes the
equivalence between power diagrams and regular polyhedral subdivisions. To see that each power diagram $\mathcal{P}$ corresponds to a regular polyhedral subdivision, consider the function 
\[p(x):= - \min_i \{g_i(x)-\norm{x}^2\} = - \min_i \{\norm{s_i}^2 -2 x\cdot s_i - w_i\}.\] This function is convex and its restriction to any cell of the power diagram is affine. Moreover, if $p$ is affine on $B\subseteq X$ then there is $P\in \mathcal{P}$ with $B\subseteq P$. Therefore, each power diagram corresponds to a regular polyhedral subdivision and we will use this observation below.

\begin{figure}
    \centering
    \includegraphics[width=1\linewidth]{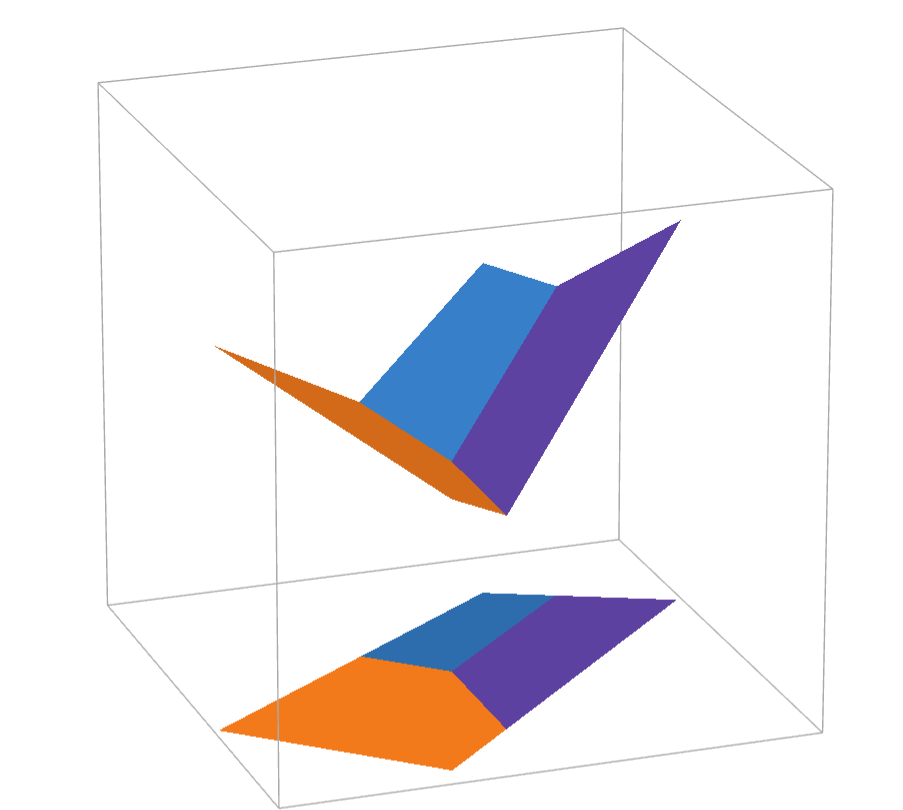}
    \caption{The lower faces of a polyhedron in \(\mathbb{R}^3\) and the corresponding polyhedral subdivision.}
    \label{fig:subdivision}
\end{figure}

\subsection{The Main Result}

\begin{proposition}
\label{prop:characterization_strongly_exposed} Let $X\subseteq \mathbb{R}%
^{n} $ be compact and convex, and let $\mu $ be a probability measure with
full support on $X$ that is absolutely continuous with respect to the
Lebesgue measure.

Let $\nu \in F_{\mu }$ have finite support. Then $\nu $ is a
Lipschitz-exposed point of $F_{\mu }$ if and only if there exists a power
diagram $\mathcal{P}$ of $X$ such that $\lambda |_{P}\preceq \mu |_{P}$ for
all $P\in \mathcal{P}$ and $\supp(\lambda )\subseteq \supp(\nu )$ jointly
imply $\lambda =\nu $.
\end{proposition}

In particular, $\nu\in F_{\mu}$ is a Lipschitz-exposed point if there is a
power diagram such that, for each cell $P$, $\nu|_P\preceq \mu|_P$ and the
support of $\nu|_P$ is affinely independent: in that case, since $%
\lambda_P\preceq \mu|_P$ implies that $\nu|_P$ and $\lambda|_P$ have the
same barycenter, affine independence of the support implies that $\nu|_P$
and $\lambda|_P$ must be equal.

\begin{proof}[Proof of Proposition \ref{prop:characterization_strongly_exposed}]
$\left( \Leftarrow \right) $ Since $\mathcal{P}$ is a power diagram, there exists a continuous convex function $%
p\colon X\rightarrow \mathbb{R}$ such that for each $P\in $\thinspace\ $%
\mathcal{P\ }$the restriction of $p$ to $P$ is affine, and if $p$ is affine
on $B\subseteq X$ then there is $P\in \mathcal{P}$ with $B\subseteq P$.

Define the Lipschitz-continuous function $u\colon X\rightarrow \mathbb{%
R}$ by 
\begin{equation*}
u\left( x\right) =p\left( x\right) -\inf_{y\in \supp(\nu )}\lVert x-y\rVert 
\text{.}
\end{equation*}%
Note that, by definition, $u\left( x\right) =p\left( x\right) $ if $x\in %
\supp(\nu )$ and $u\left( x\right) <p\left( x\right) $ for $x\in X\setminus %
\supp(\nu )$. We now claim that $\nu $ is the unique solution to 
\begin{equation*}
\max_{\lambda \in F_{\mu }}\int u(x)d\lambda (x)  \tag{L}
\label{eq:util_maxapp}
\end{equation*}

\noindent \textbf{Step 1:} $\nu $\textit{\ is a solution to problem %
\eqref{eq:util_maxapp}.}

By construction, $\int u(x)d\nu (x)=\int p\left( x\right) d\nu (x)$.
Moreover, for each $P\in \mathcal{P}$, $\int_{P}p\left( x\right) d\nu
=\int_{P}p\left( x\right) d\mu $ since $\nu |_{P}\in F_{\mu |_{P}}$ and $p$
is affine on $P$. Therefore, 
\begin{equation*}
\int u(x)d\nu (x)=\int p\left( x\right) d\mu (x).  \tag{\(1\)}
\label{eq:unu_pmu}
\end{equation*}

For any $\lambda \in F_{\mu }$ 
we obtain 
\begin{equation*}
\int u(x)d\lambda (x)\leq \int p\left( x\right) d\lambda (x)\leq \int
p\left( x\right) d\mu (x)\text{,}
\end{equation*}%
where the first inequality follows from $u\le p$ and the second follows by the definition of the convex order.
We conclude that $\nu $ is a solution to \eqref{eq:util_maxapp}.

\medskip

\noindent \textbf{Step 2:} \textit{There is no other solution to problem %
\eqref{eq:util_maxapp}.}\footnote{%
For a similar argument, see the proof of Theorem 2 in \cite%
{arieli2023optimal}.}

Let $\lambda \in F_{\mu }$ solve problem \eqref{eq:util_maxapp}. First, we
claim that $\supp\lambda \subseteq \supp\nu $. If not, $\int u(x)d\lambda
(x)<\int p\left( x\right) d\lambda (x)$ since $u(x)<p\left( x\right) $ for
all $x\not\in \supp\nu $. This implies 
\begin{equation*}
\int u(x)d\lambda (x)<\int p\left( x\right) d\lambda (x)\leq \int p\left(
x\right) d\mu (x)=\int u(x)d\nu (x)\text{,}
\end{equation*}%
where the second inequality follows because $p$ is convex and because $%
\lambda \in F_{\mu }$, and where the equality follows from \eqref{eq:unu_pmu}%
. We conclude that $\lambda $ does not solve problem \eqref{eq:util_maxapp},
a contradiction.

We now establish that $\lambda |_{P}\preceq \mu |_{P}$ for all $P\in 
\mathcal{P}$: By Bauer's Theorem, we can assume that $\lambda $ is an
extreme point of $F_{\mu }$. By Proposition \ref{prop:necessity_extreme},
there exists a finite collection $\mathcal{Q}=\{Q_{j}|j\in J\}$ of convex
sets such that, for all $Q\in \mathcal{Q}$, $\mu (Q)>0$, $\lambda |_{Q}\prec
\mu |_{Q}$ and the support of $\lambda $ on $Q$ is affinely independent.
Moreover, we can assume that the partition $\mathcal{Q}$ cannot be further
refined. Note then that 
\begin{equation*}
\int p\left( x\right) d\lambda \leq \int p\left( x\right) d\mu =\int
u(x)d\nu =\int u(x)d\lambda \leq \int p\left( x\right) d\nu ,
\end{equation*}%
where:

\begin{enumerate}
\item the first inequality follows since $p$ is convex and since $\lambda
\in F_{\mu }$,

\item the first equality follows from \eqref{eq:unu_pmu},

\item the second equality follows since $\lambda $ achieves, by assumption,
the same objective value as $\nu $

\item the final inequality follows since $u\leq p$.
\end{enumerate}

Therefore, both inequalities are, in fact, equalities. In particular, for
each $Q\in \mathcal{Q}$, $\int_{Q}p\left( x\right) d\lambda
(x)=\int_{Q}p\left( x\right) d\mu (x)$. Since $Q$ has positive $\mu $%
-measure and since the partition $\mathcal{Q}$ cannot be refined, the
function $p$ must be affine on $\ch(\supp(\mu |_{Q}))$.

It follows that for each $Q\in \mathcal{Q}$ there is $P\in \mathcal{P}$ such
that $\ch(\supp(\lambda |_{Q}))\subseteq \ch(\supp(\mu |_{Q}))\subseteq P$.
Therefore, $\lambda |_{P}\prec \mu |_{P}$ for each $P\in P,$ and our
hypothesis implies that $\lambda =\nu $. \bigskip

\noindent $\left( \Rightarrow \right) $ Let $\nu $ be the unique solution to %
\eqref{eq:util_maxapp}, where $u$ is Lipschitz-continuous. Proposition 4 in 
\cite{dworczak2019persuasion} implies that the associated moment persuasion
problem with prior $\mu \ $is regular. It follows from Theorem 3 in \cite%
{dworczak2019persuasion} that there is a convex function $p:X\rightarrow \mathbb{R%
}$ with $p\geq u$ that satisfies 
\begin{equation*}  \label{eq:udnu=pdmu}
\int u(x)d\nu (x)=\int p\left( x\right) d\mu (x)\text{.}  \tag{\(2\)}
\end{equation*}

This implies, in particular, that 
\begin{equation*}  \label{eq:u_=_p}
u(x)=p\left(x\right) \text{ for all } x\in \supp(\nu)\text{.}  \tag{\(3\)}
\end{equation*}

Since $\nu $ is the unique solution to problem \eqref{eq:util_maxapp},
Bauer's Theorem implies that $\nu $ is an extreme point of $F_{\mu }$.
Proposition \ref{prop:necessity_extreme} implies that there exists a
partition $\mathcal{P}$ of $X$ such that, for each $P\in \mathcal{P}$, $P$
is convex, $P$ has positive $\mu $-measure, $\nu |_{P}\prec \mu |_{P}$, and $%
\supp(\nu |_{P})$ is affinely independent. Assume that $\mathcal{P}$ is a
finest partition with these properties.

We claim that, for each $P\in \mathcal{P}$, $p$ is affine on $P$. Let $%
\{x_{1},...,x_{m}\}$ be the support of $\nu |_{P}$ and let $\mu
|_{P}=\sum_{i}\mu _{i}$ be such that $\nu (\{x_{i}\})=\mu _{i}(X)$ and $%
r(\mu _{i})=x_{i}$ (see Lemma \ref{lemma:cartier}). Note that since $p$ is
convex and since $\int p\left( x\right) d\nu |_{P}=\int p\left( x\right)
d\mu |_{P}$, we obtain that $\int p\left( x\right) d\mu _{i}(x)=p(x_{i})$.
Jensen's inequality implies that $\int p\left( x\right) d\mu
_{i}(x)=p(x_{i}) $ if and only if $p$ is affine on $\ch(\supp(\mu _{i}))$.
We conclude that $p $ is affine on $P$ since otherwise we could refine the
partition $\mathcal{P}$, contradicting our assumption.

The above argument implies that $p$ is a piecewise affine, convex function.
Let $\mathcal{Q}$ denote the power diagram corresponding to the lower
envelope of this function (recall Aurenhammer's fundamental equivalence
Theorem). Suppose that there is $\lambda \neq \nu $ satisfying $\lambda
|_{Q}\prec \mu |_{Q}$ for all $Q\in \mathcal{Q}$ and $\supp(\lambda
)\subseteq \supp(\nu )$. Using \eqref{eq:u_=_p}, this yields $\int
u(x)d\lambda (x)=\int p\left( x\right) d\lambda (x)$. Since $p$ is affine on 
$Q$ and since $\lambda |_{Q}\prec \mu |_{Q}$, we obtain for any $Q\in 
\mathcal{Q}$ that $\int p\left( x\right) d\lambda |_{Q}(x)=\int p\left(
x\right) d\mu |_{Q}(x)$. Putting these observations together and using %
\eqref{eq:udnu=pdmu}, we finally obtain 
\begin{equation*}
\int u(x)d\lambda (x)=\int p\left( x\right) d\lambda (x)=\int p\left(
x\right) d\mu (x)=\int u(x)d\nu (x)
\end{equation*}%
This contradicts the hypothesis that $\nu $ is the unique solution to %
\eqref{eq:util_maxapp}, and therefore we conclude that $\lambda |_{Q}\prec
\mu |_{Q}$ for all $Q\in \mathcal{Q}$ and $\supp(\lambda )\subseteq \supp%
(\nu )$ imply $\lambda =\nu $.
\end{proof}

The result above treats extreme fusions with finite support. It is clear
that the original measure $\mu $ is itself an extreme point of $F_{\mu }$. In a persuasion problem, this corresponds to the signal that perfectly
reveals to the receiver each realized state. Perfect revelation can be
optimal in some cases, i.e., if the sender's utility function is convex
on a certain domain. The next Corollary shows that a combination of perfect
revelation and a finitely supported fusion can also constitute a
Lipschitz-exposed point. We first need a Lemma whose proof
uses the theory of the Monge-Ampere equation: 

\begin{lemma}
\label{l:existence_strictly_convex} For any compact and convex set $X\subset 
\mathbb{R}^{n}$ there is a Lipschitz-continuous and convex function $%
f:X\rightarrow \mathbb{R}$ that is strictly convex on the interior of $X$
and that satisfies $f(x)=0$ for all $x$ in the boundary of $X$.
\end{lemma}

\begin{proof}
    It follows from Theorem 1.1 in \citet{hartenstine} that the Dirichlet problem of the Monge-Ampere equation $det D^2f=\mu$ subject to the constraint $f=0$ on the boundary of $X$ has a continuous convex solution $f$ for any finite Borel measure $\mu$. By Corollary 5.2.2 in \citet{gutierrez}, this solution is strictly convex on the interior of $X$ whenever $\mu$ is strictly positive. Moreover, $f$ is Lipschitz-continuous by Theorem 5.4.8 in \citet{gutierrez}.
\end{proof}

\begin{corollary}
\label{cor:suff_for_exposed} Let $X\subseteq \mathbb{R}^{n}$ be compact and
convex, and let $\mu $ be a probability measure with full support on $X$
that is absolutely continuous with respect to the Lebesgue measure. Let $%
\upsilon $ be a fusion of $\mu $ and suppose that there exists a power
diagram $\mathcal{P}$ of $X$ such that, for each $P\in \mathcal{P}$, either $%
\nu |_{P}\preceq \mu |_{P}$ and $\nu |_{P}$ has affinely independent
support, or $\nu |_{P}=\mu |_{P}$. Then $\nu $ is a Lipschitz-exposed point
of $F_{\mu }$.
\end{corollary}

The proof of this Corollary (in the Appendix) follows the proof of one
direction in Proposition \ref{prop:characterization_strongly_exposed}, but
adjusts the function $p$ that we construct to be strictly convex on some
cells using Lemma \ref{l:existence_strictly_convex}. 

One can construct Lipschitz objective functionals in which the unique optimal fusion is induced by a partition of \(X\) into an uncountable collection of convex sets, on each of which \(\mu\) is collapsed to its barycenter. A natural open question, therefore, is whether all Lipschitz-exposed points take the form of a (possibly infinite) collection \(\mathcal{P}\) of convex sets (satisfying some regularity condition) on each of which either \(\nu|_P \preceq \mu|_P\) and \(\nu|_P\) has affinely independent support, or \(\nu|_P = \mu|_P\). One challenge is in the precise formalization of these objects, as the mathematical literature primarily studies the case where \(\mathcal{P}\) has countably many elements. Perhaps a clever approximation approach could work.


\subsection{The Gap Between Necessity and Sufficiency}


We now note that there is a small gap between the necessary condition for extreme points identified in Proposition \ref{prop:necessity_extreme} and our sufficient condition for Lipschitz-exposed points. To that end, we present an example illustrating that the necessary condition cannot be strengthened to the partition being a power diagram. We also present an example of an extreme point that is not Lipschitz-exposed. One perspective on the gap is that it is not too large: Straszewicz’s Theorem (Theorem 18.6 in \cite{rockafellar1970}) and Klee's Theorem (\cite{klee1958extremal}) state that the set of exposed points of a compact convex set is, under mild conditions, dense in the set of extreme points; and by forming the convolution of an integrable function with an appropriate smooth mollifier, one can approximate an integrable function with a Lipschitz-continuous one.

Our first example shows that there are Lipschitz-exposed points (which are
therefore extreme points) for which there is no power diagram such that the
support on each cell is affinely independent. This shows that the necessary
conditions derived in Proposition \ref{prop:necessity_extreme} cannot be
strengthened to require that the partition $\P $ of $X$ into convex sets is
a power diagram and the support on each cell is affinely independent.

\begin{example}
Let $\mu $ be the uniform distribution on the rectangle $\left[ 0,2\right]
\times \left[ 0,1\right] $, and let $\nu $ be the fusion obtained by
contracting $\mu $ on $\left[ 1,2\right] \times \left[ 0,1\right] $ to $%
\left( \frac{3}{2},\frac{1}{2}\right) $; on $\left[ 0,1\right] \times \left[
0,\frac{1}{2}\right] $ to $\left( \frac{3}{8},\frac{1}{4}\right) $ and $%
\left( \frac{5}{8},\frac{1}{4}\right) $; and on $\left[ 0,1\right] \times %
\left[ \frac{1}{2},1\right] $ to $\left( \frac{3}{8},\frac{3}{4}\right) $
and $\left( \frac{5}{8},\frac{3}{4}\right) $. This is depicted in Figure \ref%
{fig:example1}.

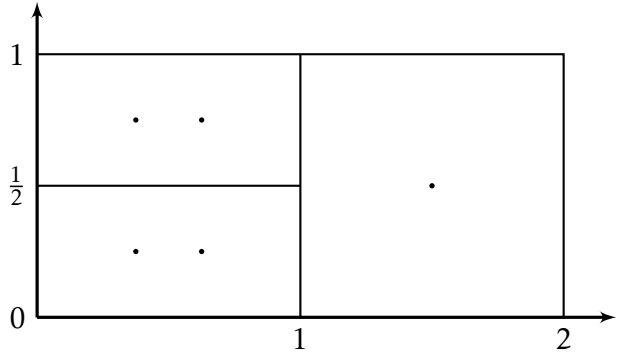
\begin{figure}[tbp]
\centering
\begin{tikzpicture}[x=3.5cm,y=3.5cm]
  \draw[axis,->] (0,0) -- (2.2,0) ;
  \draw[axis,->] (0,0) -- (0,1.2) node[above] {};
  \draw[thick] (0,0) -- (2,0) -- (2,1) -- (0,1) -- (0,0);
  \draw[thick] (1,0) -- (1,1);
  \draw[thick] (0,.5) -- (1,.5);
  \draw (0,0) node[left] {$0$};
  \draw (0,0.5) node[left] {$\frac{1}{2}$};
  \draw (0,1) node[left] {$1$};
  \draw (1,0) node[below] {$1$};
  \draw (2,0) node[below] {$2$};
  \fill (1.5,0.5)  circle[radius=1pt];
  \fill (3/8,0.75) circle[radius=1pt];
  \fill (5/8,0.75) circle[radius=1pt];
  \fill (3/8,0.25) circle[radius=1pt];
  \fill (5/8,0.25) circle[radius=1pt];
\end{tikzpicture}
\caption{Example 1}
\label{fig:example1}
\end{figure}

Observe that the partition of $X$ shown in Figure \ref{fig:example1} is not a polyhedral
subdivision, as the intersection of the upper small rectangle with the large rectangle is not a common face. Moreover, there is no polyhedral subdivision 
$\mathcal{P}$ of $X$ such that, for each $P\in \mathcal{P}$, $\nu
|_{P}\preceq \mu |_{P}$ and the support of $\nu |_{P}$ is affinely
independent.

Nonetheless, one can verify using Proposition \ref%
{prop:characterization_strongly_exposed} that $\nu $ is a Lipschitz-exposed
point of $F_{\mu }$.\footnote{%
To see this, let $p\colon \left[ 0,2\right] \times \left[ 0,1\right]
\rightarrow \mathbb{R}$ be given by $p\left( x\right) =0$ if $x_{1}\leq 1$
and $p\left( x\right) =x_{1}-1$ if $x_{1}>1$. Let 
\begin{equation*}
u\left( x\right) =p\left( x\right) -\inf_{y\in \supp(\nu )}\lVert x-y\rVert 
\text{.}
\end{equation*}%
Since $\int u\left( x\right) \,\mathrm{d}\nu \left( x\right) =\int p\left(
x\right) \,\mathrm{d}\mu \left( x\right) $, $\nu $ solves \ref{eq:util_maxapp}.
Finally, there is no other distribution $\lambda $ such that $\supp\lambda
\subseteq \supp\nu $ and such that $\lambda |_{\left[ 0,1\right] \times %
\left[ 0,1\right] }\preceq \mu |_{\left[ 0,1\right] \times \left[ 0,1\right]
}$ and $\lambda |_{[1,2]\times \left[ 0,1\right] }\preceq \mu |_{[1,2]\times %
\left[ 0,1\right] }$. Suppose there was such a $\lambda $. By symmetry of
the left part of the figure, we could obtain another fusion by mirroring $%
\lambda $ along a horizontal or a vertical line. The convex combination of
the four fusions we obtain this way would equal $\nu $, implying that $\nu $
is not an extreme point. This contradicts the conclusion of Lemma \ref%
{lemma:halfspaces_extreme_points} in the Appendix.}
\end{example}

We now derive weaker conditions that are necessary for a fusion to be the
unique solution. 

\begin{corollary}
\label{cor:convexly_independent} Let $\mu $ be absolutely continuous and
have full support. Suppose $\nu $ has finite support and is a
Lipschitz-exposed point of $F_{\mu }$. Then there is a Power diagram $%
\mathcal{Q}$ such that the support of $\nu |_{Q}$ is convexly independent
for all $Q\in \mathcal{Q}$.
\end{corollary}

We use this result in the following example to illustrate a fusion $\nu$
that is an extreme point of $F_{\mu}$ but not an exposed point.

\begin{example}
Let $\mu $ be the uniform distribution on the rectangle $\left[ 0,2\right]
\times \left[ 0,1\right] $, and let $\nu $ be the fusion obtained by
contracting $\mu $ on $\left[ 1,2\right] \times \left[ 0,1\right] $ to $%
\left( \frac{3}{2},\frac{1}{2}\right) $; on $\left[ 0,1\right] \times \left[
0,\frac{1}{2}\right] $ to $\left( \frac{1}{2},\frac{1}{4}\right) $; and on $%
\left[ 0,1\right] \times \left[ \frac{1}{2},1\right] $ to $\left( \frac{1}{2}%
,\frac{7}{10}\right) $ and $\left( \frac{1}{2},\frac{8}{10}\right) $. This
is depicted in Figure \ref{fig:example2}.

The only power diagram such that $\nu|_P\prec \mu|_P$ on each cell $P$ is
the trivial power diagram with only one cell $P=X$. However, since the
support of $\nu$ is not convexly independent, Corollary \ref%
{cor:convexly_independent} implies that $\nu$ is not a Lipschitz-exposed
point. But it follows from Lemma \ref{lemma:halfspaces_extreme_points} that
it is an extreme point.
\end{example}

\begin{figure}[tbp]
\centering
\begin{tikzpicture}[x=3.5cm,y=3.5cm]
  \draw[axis,->] (0,0) -- (2.2,0) ;
  \draw[axis,->] (0,0) -- (0,1.2) node[above] {};
  \draw[thick] (0,0) -- (2,0) -- (2,1) -- (0,1) -- (0,0);
  \draw[thick] (1,0) -- (1,1);
  \draw[thick] (0,.5) -- (1,.5);
  \draw (0,0) node[left] {$0$};
  \draw (0,0.5) node[left] {$\frac{1}{2}$};
  \draw (0,1) node[left] {$1$};
  \draw (1,0) node[below] {$1$};
  \draw (2,0) node[below] {$2$};
  \fill (1.5,0.5)  circle[radius=1pt];
  \fill (.5,0.8) circle[radius=1pt];
  \fill (.5,0.7) circle[radius=1pt];
  \fill (.5,0.25) circle[radius=1pt];
\end{tikzpicture}
\caption{Example 2}
\label{fig:example2}
\end{figure}
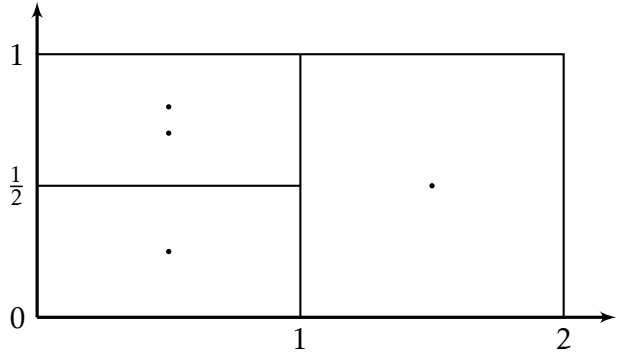

\section{Convex partitional fusions and moment persuasion}\label{section4}

We say that a fusion $\nu $ of $\mu $ is \emph{convex partitional} if there
is a partition of $X$ into convex sets such that, for each cell $P$, the
support $\nu |_{P}$ is a singleton and $\nu |_{P}\prec \mu |_{P}$. An
extreme point of $F_{\mu }$ need not be convex partitional, but convex
partitional fusions are special in the sense that they are the finest
fusions among those with a fixed number of support points. The next
Proposition generalizes the one-dimensional result (Lemma 1) of \cite%
{ivanov2021optimal}:

\begin{proposition}
\label{prop:convexpartfiner} Suppose that $\nu $ is a fusion of $\mu $ with $%
K$ points in its support. Then there exists a convex partitional measure $%
\lambda $ with at most $K$ points in its support that satisfies $\nu \preceq
\lambda \preceq \mu $. If $\nu $ is not convex partitional, then $\nu \prec
\lambda $. 
\end{proposition}

In Section \ref{section5}, we apply this proposition to categorization. Under processing or memory constraints in which a decision maker is constrained to arrange data in a number of finite ``bins,'' there always exists an optimal arrangement in which the data is categorized according to convex partitions of the state space.


\subsection{Moment Persuasion}\label{sec:moment_persuasion}

A special class of multidimensional Bayesian persuasion problems are those
in which the sender's payoff can be written in a reduced-form way as a
function of the receiver's vector of posterior means. That is, given $G\in
\Delta \left( \Theta \right) $ with 
\begin{equation*}
\int_{\Theta _{1}}adG_{1}\left( a\right) =x_{1}, \int_{\Theta
_{2}}adG_{2}\left( a\right) =x_{2},\dots, \int_{\Theta _{d}}adG_{d}\left(
a\right) =x_{d}\text{,}
\end{equation*}
there exists a measurable function $V \colon \Theta \to \mathbb{R}$ such
that the sender's payoff is $V(x)$, where $x = (x_1,\dots,x_d)$. We impose
that $V$ is Lipschitz on $\Theta$.

Given this reduced-form payoff $V$, the receiver solves 
\begin{equation*}
\max_{\nu \in F_{\mu }}\mathbb{E}_{\nu}[V]\text{.}
\end{equation*}%
If the receiver has only finitely many actions, i.e., if $\left\vert
A\right\vert =m\in \mathbb{N}$, then there is a solution to the persuasion
problem where the optimal distribution of posterior means $\nu^{\ast }$ has
at most $m$ points in its support.

Proposition \ref{prop:necessity_extreme} suggests then the following
interpretation of an optimal solution to a persuasion problem. For any
finitely supported extreme point $\nu $ of $F_{\mu }$, we define $\hat{\nu}%
,\ $the \emph{convex-partitional coarsening} of $\nu $, to be the fusion of $%
\nu \ $obtained by taking each element $P\in \mathcal{P}$ - the underlying
partition supporting $\nu \ $(see Proposition \ref{prop:necessity_extreme})
- and collapsing $\mu \ $on it to its barycenter, $r_P(\mu)$. By
construction, $\hat{\nu}\in F_{\nu }\subseteq F_{\mu }$ since the convex
order is transitive. Moreover, $\nu =\hat{\nu}$ if and only if $\nu $ is
itself convex partitional.

Let $D_{\mu ,P}$ denote the set of \emph{mean-preserving spreads} of $\delta
_{r_{P}\left( \mu \right) }$ that are supported on $P$. The \emph{%
unconstrained persuasion problem} on $P$ is 
\begin{equation*}
\max_{\nu \in D_{\mu ,P}}\mathbb{E}_{\nu}V\text{.}
\end{equation*}
We say that an unconstrained persuasion problem is \emph{standard} if $P$ is
a simplex. As any simplex is homeomorphic to the standard simplex, any
standard unconstrained persuasion problem is equivalent to a problem of the
class studied in \cite{kamenicagentzkow2011}.

We define the upper concave envelope of $V$ on $P$, $\cav_P V$, to be the
smallest concave function that lies pointwise above $V$ on $P$.
Alternatively, we call this the concavification of $V$ on $P$. We say that a
solution to the persuasion problem, $\nu^*$, is \emph{canonical} if on each
cell $P$ of the corresponding partition $\mathcal{P}^*$, $\nu^{*}|_{P}$ is
the solution of an unconstrained persuasion problem on $P$. Moreover, the
sender's payoff is 
\begin{equation*}
\sum_{P \in \mathcal{P}^*} \mu\left(P\right) \cav_P
V\left(r_{P}\left(\mu\right)\right)
\end{equation*}

\begin{proposition}
\label{prop:canonical} There exists a canonical solution to the persuasion
problem.
\end{proposition}

\begin{proof}
By Bauer's maximum principle, there is a solution to the persuasion problem, \(\nu^*\), that is an extreme point. As the receiver only has finitely many actions, \(\nu^*\) can be taken to have finitely many support points. Let \(\mathcal{P}^*\) be a convex partition of \(\Theta\) corresponding to \(\nu^*\) such that there is no strictly finer partition corresponding to \(\nu^*\).

Suppose that some \(P \in \mathcal{P}^*\)  
\(\nu^{*}|_{P}\) has multiple points of support (or else we are done). As there is no strictly finer partition corresponding to \(\nu^{*}\), by Proposition \ref{prop:convexpartfiner}, there exists a \(\lambda^{*}|_{P}\) such that \(\nu^{*}|_{P} \prec \lambda^{*}|_{P} \preceq \mu|_{P}\). Moreover, as \(\supp \nu^{*}|_{P}\) is affinely independent, and as \(\nu^{*}|_{P} \prec \lambda|_{P}\) implies that \(\conv \supp \nu^{*}|_{P} \subseteq \conv \supp \lambda^{*}|_{P}\), Proposition \ref{prop:convexpartfiner} implies that \(\supp \lambda^{*}|_{P}\) is affinely independent.

By construction, \(\nu^{*}|_{P}\) is a solution to \(\max_{\nu \in F_{\lambda^{*}|_{P}}}\mathbb{E}_{\nu}V\). This is a standard unconstrained persuasion problem, whose solution is given by the concavification of the persuader's value function on \(\conv \supp \lambda^{*}|_{P}\) (\cite{kamenicagentzkow2011}). Let \(\mathcal{H}^*(x)\) be the set of hyperplanes supported on \(\supp \nu^{*}|_{P}\) that lie weakly above \(V\) on \(\conv \supp \lambda^{*}|_{P}\).

Suppose for the sake of contradiction that there exists an \(x \in P\) for which \(H^*(x) < V(x)\) for all \(H^*(x) \in \mathcal{H}^*\). By Theorem 5 in Dworczak and Kolotilin, strong duality holds for multi-dimensional moment persuasion, and as \(\nu^*\) is an optimal fusion of \(\mu\) on \(\Theta\), \(\nu^{*}|_{P}\) is an optimal fusion of \(\mu|_{P}\) (on \(P\)). This implies that there exists a convex piecewise-affine function \(q_P\) that lies everywhere above \(V\) on \(P\) and such that \(q_P(y) = V(y)\) for all \(y \in \supp \nu^{*}|_P\). However, \(\int q_P d\mu|_P = \int V d\nu^{*}|_P\) only if \(q_P\) is affine on \(P\), by our assumption that there is no strictly finer partition, a contradiction. \end{proof}

This result indicates that the solution to a moment persuasion problem (in
which finitely many different actions are induced) can be decomposed into
two parts. First, in the \textbf{global} portion of the problem, the domain
is partitioned into finitely many convex sets. Second, in the \textbf{local}
portion of the problem, the persuader solves an unconstrained problem on
each of these convex sets, whose solution is given by the concavification of
the value function. As Proposition \ref{prop:canonical} states, there always
exists such a canonical solution to a moment-persuasion problem.

This result also elucidates the bi-pooling result of \cite{kleiner2021extreme}
and \cite{arieli2023optimal}: any moment-persuasion problem with an interval
state space admits a bi-pooling solution. When the optimal solution has a
finite support, this means that the interval is partitioned into finitely
many intervals such that, on each of them, the prior is either collapsed to
its barycenter, or is fused into two points. This is precisely the
single-dimensional version of Proposition \ref{prop:necessity_extreme}.

\section{An Application to Categorization}\label{section5}

Categorization, or the arrangement of objects according to some rule, is a central part of machine learning and our interaction with the world more broadly. An influential work that studies this formally is \cite{gardenfors2004conceptual}. There, a central idea is the notion of a conceptual space--in economics, a state space--and its division into convex sets, each containing a prototype, i.e., a central object. For example, one can think of colors: ``red'' describes a whole family of colors, as does ``blue'' and so on. \citeauthor{gardenfors2004conceptual} notes several ways of producing this phenomenon: one can define a \emph{natural property} as a convex region of the conceptual space, in which case the prototype in each region is the central object in each region. He also notes that one can instead start with a finite set of prototypes and then partition the conceptual space by grouping points together that are closest to each prototype. When the notion of closeness is the Euclidean distance, this corresponds to a particular type of convex partition of the conceptual space, a Voronoi tessellation.

These constructions seem reasonable and realistic, yet the central properties are exogenous. That is, convexity of the regions, or existence of the prototypes, is/are assumed. Here, we use our earlier results to argue that these properties are natural outcomes of a decision-maker's (DM's) optimization when she faces limits on the amount of information she can acquire or process. We provide two different micro-foundations for categorization. In both, the outcome is a convex partition of the state space in which there exists a single representative object (prototype) per partition element. In the first derivation, we show that this categorization emerges as information becomes sufficiently cheap when a decision-maker (DM) acquires information. In the second derivation, we show that this phenomenon emerges when a DM must instead store information in finitely many bins. We show that the optimal way of doing so is precisely categorization, i.e., by classifying information into convex bins.

\subsection{Decision Problems with Flexible Information Acquisition}

Let the state of the world $x \in X \equiv \left[0,1\right]^{d}$ be distributed according to absolutely continuous \(\mu\). There is a risk-neutral DM with a bounded utility function, who must choose action \(a\) from some finite set of undominated actions \(A\), $\left|A\right| = t$ (\(\infty > t \geq 2\)). As the maximum of Lipschitz-continuous functions is Lipschitz-continuous, the agent's decision problem induces a Lipschitz-continuous (piecewise-affine) reduced-form value function, \(V\), of the posterior mean.

Prior to her decision, the DM acquires (or processes) information. We assume that this can be modeled as follows. There is a cost functional $C\colon F_\mu \to \mathbb{R}$, where
\[C\left(\nu\right) = \kappa \int c d\nu \text{ ,}\]
with \(\kappa > 0\) and \(c\) some strictly convex, Lipschitz-continuous, twice continuously differentiable function whose partial derivatives are bounded on \(X\); i.e., acquiring any fusion \(\nu \in F_{\mu}\) costs the DM $C\left(\nu\right)$. The DM solves
\[\max_{\nu \in F_\mu} \int (V-\kappa c)d\nu \text{.}\]

\begin{remark}\label{remark:cheapcat}
For any decision problem, cost function, and prior of the form described above, there exists a \(\bar{\kappa} > 0\) such that if the cost parameter satisfies \(\kappa \leq \bar{\kappa}\), there exists a solution to the DM's information acquisition problem that corresponds to categorization with a single prototype per category.
\end{remark}
\begin{proof}
Fix a decision problem, cost function, and prior of the form described above. Let \(c_i\) denote the partial derivative of \(c\) in \(x_i\): \(c_i(x) \equiv \frac{\partial}{\partial{x_i}} c\). Define \[\alpha \equiv \max_{x, y \in \left[0,1\right]^{d}; i = 1, \dots, d}{\left|c_i\left(x\right) - c_i\left(y\right)\right|} \text{ ,}\]
i.e., \(\alpha\) is the upper bound for the largest possible (absolute) difference between marginal costs of beliefs. Because the partial derivatives of \(c\) are bounded and \(c\) is strictly convex, \(\alpha \in \mathbb{R}_{++}\). Similarly, define 
\[\beta \equiv \min_{x, y \in \left[0,1\right]^{d}; i = 1, \dots, d}{\left\{\left| V_i\left(x\right) - V_i\left(y\right)\right| \ \colon \ \left| V_i\left(x\right) - V_i\left(y\right)\right| > 0\right\}} \text{ .}\]
Because there are only finitely many actions and all are undominated, \(\beta \in \mathbb{R}_{++}\). Finally, let \(\kappa < \beta/\alpha\).

By Proposition 4, there exists a canonical solution, \(\nu\), to the persuasion problem. Suppose for the sake of contradiction that on one of the elements \(P \in \mathcal{P}\), \(\nu\) has multiple (\(m \geq 2\)) points of support. 
This collection of points $\left(x_1, \dots, x_m\right)$ must be such that for each $j = 1, \dots, d$, and for any $i, k \in \left\{1,\dots,m\right\}$, 
\[V_j\left(x_i\right) - V_j\left(x_k\right) = \kappa\left(c_j\left(x_i\right) - c_j\left(x_k\right)\right) \text{ .}\]
By the optimality of \(\nu\) and the strict convexity of \(c\), because \(\nu\) has multiple support points on \(P\), the DM must take at least two different actions with strictly positive probability. Therefore, for at least one trio $j, i, k$ the left-hand side of this equation is nonzero and so
\[\kappa = \frac{V_j\left(x_i\right) - V_j\left(x_k\right)}{c_j\left(x_i\right) - c_j\left(x_k\right)} \geq \frac{\beta}{\alpha}\text{,}\]
a contradiction.\end{proof}

\subsection{Decision Problems with Finite Memory}

We now drop the assumption that the DM's problem induces a mean-measurable value function\footnote{By the Stone-Weierstrass theorem, this assumption earlier is innocuous.} and also jettison the specification that the cost of information acquisition is smooth. Instead, the DM has access to \(K \in \mathbb{N}\) ``bins'' or categories, and she chooses a stochastic mapping that assigns states to bins. When choosing her action $a$, the DM learns only to which bin the realized state was mapped, not the actual state. Importantly, we do not assume that the \(K\) bins correspond to a partition, much less a convex one. For instance, the DM could assign any \(x \in X\) to one of the bins (uniformly) at random. The DM's problem becomes
\begin{align*}
    &\max_{\nu\in F_{\mu}} \int V\,\mathrm d\nu\\
    &\text{s.t. } |\supp \nu|\le K.
\end{align*} 

Since information is always beneficial to a DM, Proposition \ref{prop:convexpartfiner} can be applied to conclude the following:
\begin{remark}
    There is an optimal categorization that is convex partitional. If the number of undominated actions in the decision problem is weakly greater than the number of possible categories, any optimal categorization must be convex partitional.
\end{remark}

\subsection{Decision-Making With Complexity Constraints}

Suppose there is a finite state space \(\Theta\). There is a DM with a compact set of actions \(A\) and a continuous utility function \(u \colon A \times \Theta \to \mathbb{R}\). The Bayesian DM has a prior, \(\eta_0 \in int \Delta\left(\Theta\right)\), and observes information before taking a decision. Formally, she observes the realization of a signal \(\pi \colon \Theta \to \Delta\left(S\right)\), where \(S\) is a compact set of signal realizations, before choosing an action \(a \in A\). We let \(\mu \in \Delta \Delta\left(\Theta\right)\) denote the distribution over posteriors \(\eta\) induced by the prior and the signal.

We specify that the DM is constrained in the following sense: there exists some \(K \in \mathbb{N}\) and the DM is constrained to optimize over decision rules \(\sigma \colon S \to \Delta\left(A\right)\) with support of at most \(K\) points. We call such a decision rule a simplicity-constrained decision rule. It is immediate that this is equivalent to the DM being restricted to choose a fusion of \(\mu\), \(\nu\), that is supported on at most \(K\) points before choosing an arbitrary (unconstrained) decision rule. We say that a simplicity-constrained decision rule is convex partitional if \(\nu\) is a convex-partitional fusion of \(\mu\). Then, Proposition \ref{prop:convexpartfiner} gives us
\begin{remark}
    There is an optimal simplicity-constrained decision rule that is convex partitional. If the number of undominated actions in the decision problem is weakly greater than the number of support points of \(\mu\), any optimal simplicity-constrained decision rule must be convex partitional.
\end{remark}


\appendix

\section{Omitted Proofs}\label{proofs}

\subsection{Proof of Proposition \protect\ref{prop:necessity_extreme}}



We first need several auxiliary Lemmas:

\begin{lemma}
\label{lemma:cartier} (\cite{cartier1964comparaison}): Let $\mu ,\nu $ be
positive measures on a compact, convex set $X$ such that $\nu
=\sum\limits_{i=1}^{n}\alpha _{i}\delta _{x_{i}}$ where$\
\sum\limits_{i=1}^{n}\alpha _{i}=1\ $\ and $\alpha _{i}>0$ for each $i.$
Then $\nu \preceq \mu $ if and only if there exist positive measures $\mu
_{1},\dots ,\mu _{n}$ such that, for all $i,$ the barycenter of $\mu _{i}$
equals $x_{i},$ and such that $\mu =\sum\limits_{i=1}^{n}\alpha _{i}\mu _{i}$
( in particular $\delta _{x_{i}}\preceq \mu _{i}$ for all $i$).
\end{lemma}

For any two measures $\mu$ and $\nu$ on $X$, $\mu\le \nu$ denotes the
pointwise order defined by $\mu(A)\le \nu(A)$ for all measurable $A\subseteq
X$.

\begin{lemma}
\label{lemma:new_measure} Let $\mu $ be a positive measure supported on $X$
and let $y\in \interior(\ch(\supp(\mu )))$. Then there exists a positive
measure $\pi $ such that $\pi \leq \mu $, $\pi (X)>0$, and $r(\pi )=y$.
\end{lemma}

\begin{proof}
Choose a finite set $\{y_1,\ldots,y_k\}\subseteq \supp \mu$ such that $y$ lies in the interior of the convex hull of $\{y_1,...,y_k\}$.
Since the convex hull operator is continuous (e.g., \cite{sertel1989continuity}), there is $\e>0$ such that for all $z_i\in B_{\e}(y_i)$, $y\in \ch(z_1,...,z_k)$. Choosing $\e>0$ small enough, we can assume that $B_{\e}(y_i)$ is disjoint from $B_{\e}(y_j)$ whenever $i\neq j$. 

For each $i$, $\mu(B_{\e}(y_i))>0$ since $y_i\in \supp(\mu)$. Moreover, the barycenter of $\mu|_{B_{\e}(y_i)}$ is in $B_{\e}(y_i)$. Therefore, there are $\lambda_i\ge 0$ with $\sum_i \lambda_i=1$ such that the barycenter of $\pi \coloneqq \sum_i \lambda_i \mu|_{B_{\e}(y_i)}$ is $y$. Since $\pi\le \mu$ by construction, the claim follows. 
\end{proof}

\begin{remark}
Observe that the measure $\pi $ constructed in the proof of the above Lemma
satisfies $\pi (X)\geq \min_{i}\{\mu (B_{\varepsilon }(y_{i}))\}$.
\end{remark}

\begin{lemma}
\label{lemma:shifting} Let $\mu _{1},\mu _{2}$ be positive measures on $X$
satisfying 
\begin{equation*}
C=\interior(\ch(\supp(\mu _{1})))\cap \interior(\ch(\supp(\mu _{2})))\neq
\emptyset \text{.}
\end{equation*}%
For all $d\in \mathbb{R}^{n}$ there is $\bar{\varepsilon}>0$ such that for
all $\varepsilon \in (0,\bar{\varepsilon})$ and $a\in \lbrack -\varepsilon
,\varepsilon ]$ there is a measure $\pi $ with $\mu _{1}+\pi \geq 0$, $\mu
_{2}-\pi \geq 0$, $\pi (X)=a$ and $\int x \,\mathrm{d}\pi =\varepsilon d$.
\end{lemma}

The Lemma implies that for any positive measures whose supports intersect
nontrivially, and for any vector $d\in\mathbb{R}^n$, we can shift mass from
one measure to the other such that the barycenters of the measures moves in
direction $d$ and $-d$ respectively, while the measures remain positive.

\begin{proof}
Let $\delta>0$ and $z\in C$ be such that $B_{2\delta}(z)\subseteq C$. Note that by Lemma~\ref{lemma:new_measure} and the following remark, there is $M>0$ such that for all $z'\in B_{\delta}(z)$ there are positive measures $\xi\le \mu_1$ and $\zeta\le \mu_2$ with barycenter $r(\xi)=r(\zeta)=z'$ and $\xi(X),\zeta(X)\ge M$.

Let $\bar{\e}= \min\left\{\frac{M}{2},\frac{\delta M}{\norm{d}}, \frac{\delta M}{\norm{z}+\delta}\right\}>0$ and fix $\e\in (0,\bar{\e})$. Suppose $a\in[0,\varepsilon]$ (the arguments for $a\in [-\varepsilon,0)$ are analogous).
Since $\frac{\e}{M}\norm{d}<\delta$, the preceding paragraph implies that there is a positive measure $\pi_1\le \mu_2$ with barycenter $z+\frac{\e}{M} d$ and mass $\pi_1(X)\ge M$. By multiplying $\pi_1$ with a positive constant less than 1, we can assume that $\pi_1(X)= M$. Analogously, since $\left(\frac{M}{M-a}-1\right)\norm{z}< \delta$, there is a positive measure $\pi_2\le \mu_1$ with barycenter $\frac{M}{M-a}z$ and mass $\pi_2(X)\ge M$. By multiplying $\pi_2$ with a positive constant less than 1, we can assume that $\pi_2(X)=M-a>0$.

 Now define $\pi\coloneqq\pi_1-\pi_2$. Then
\[\mu_1 + \pi = \mu_1 + \pi_1 - \pi_2 \ge \pi_1 \ge 0,\]
where the first inequality follows from $\pi_2 \le \mu_1$, and the second holds since $\pi_1$ is positive. Similarly, because $\pi_1 \le \mu_2$ and $\pi_2$ is positive, we have $\mu_2 - \pi \ge 0$. Moreover, $\pi(X)=a$ and
 \[ \int x \,\mathrm d\pi = \int x \,\mathrm d\pi_1 - \int x \,\mathrm d\pi_2 = M\left(z+\frac{\e}{M} d\right)- (M-a)\frac{M}{M-a}z= \e d, \]
 which establishes the claim.
 \end{proof}

\begin{corollary}
\label{lemma:start} Let $\mu_1,\mu_2$ be positive measures with barycenters $%
x_1$ and $x_2$, respectively. If 
\begin{equation*}
C=\interior(\ch(\supp(\mu_1)))\cap \interior(\ch(\supp(\mu_2)))\neq \emptyset%
\text{,}
\end{equation*}

\begin{enumerate}
\item there is a measure $\pi $ satisfying $\pi (X)>0$, $\mu _{1}+\pi \geq 0$%
, $\mu _{2}-{\pi }\geq 0$ and such that the barycenter of $\tilde{\mu}_{1}%
\coloneqq\mu _{1}+\pi $ is $x_{1}$.

\item there is a measure $\pi $ satisfying $\pi (X)<0$, $\mu _{1}+\pi \geq 0$%
, $\mu _{2}-{\pi }\geq 0$ and such that the barycenter of $\tilde{\mu}_{1}%
\coloneqq\mu _{1}+\pi $ is $x_{1}$.
\end{enumerate}
\end{corollary}

The Corollary establishes that if two positive measures have supports that
intersect nontrivially, we can move mass from one measure to the other
without changing the barycenters of the measures.

\begin{proof}
By Lemma \ref{lemma:shifting}, for $\e>0$ small enough there exists a measure $\pi$ with $\mu_1+\pi\ge 0$, $\mu_2-\pi\ge 0$, $\pi(X)=\e$ and $\int x \,\mathrm d\pi= \e x_1$. It follows that $r(\tilde{\mu}_1)= \frac{1}{\mu_1(X)+\e}[\mu_1(X) x_1 + \e x_1]=x_1$. The argument for the second part is analogous.
\end{proof}

For a set $A\subseteq \mathbb{R}^n$, let $\aff A$ denote the affine span of $%
A$: 
\begin{equation*}
\aff A = \{\sum_{i=1}^k \alpha_i x_i: k>0, x_i\in A, \alpha_i\in\mathbb{R},
\sum_{i=1}^k \alpha_i=1\}.
\end{equation*}

\begin{lemma}
\label{lemma:decompo} Let $\{\mu_i|i\in I\}$ be a finite collection of
positive measures with $r(\mu_i)=x_i$ and let $\mu= \sum_{i\in I}\mu_i$. Let 
$V=\{J_{1},...,J_{k}\}$ be a partition of $I$.

Consider the graph $G$ whose vertex set is $V$ and which contains an edge
between $J$ and $J^{\prime }$ if 
\begin{equation*}
\interior(\ch(\supp(\sum_{j\in J}\mu _{j})))\cap \interior(\ch(\supp%
(\sum_{j\in J^{\prime }}\mu _{j})))\neq \emptyset .
\end{equation*}%
Suppose that $G$ is connected, and let $\lambda $ be a measure such that $%
\mu +\lambda $ is positive and $r(\mu +\lambda )\in \aff\{x_{i}|i\in I\}$.

Then, for all $\varepsilon >0$ small enough we can decompose $\tilde{\mu}%
\coloneqq\mu +\varepsilon \lambda $ into positive measures $\tilde{\mu}_{J}$
with $\sum_{J\in V}\tilde{\mu}_{J}=\tilde{\mu}$ and $r(\tilde{\mu}_{J})\in %
\aff\{x_{j}|j\in J\}$.
\end{lemma}

\begin{proof}
Let $\mu_J=\sum_{i\in J}\mu_i$, and let $\{\lambda_{J}\}_{J\in V}$ be such that $\lambda=\sum_{J\in V} \lambda_{J}$, $\hat{\mu}_{J}\coloneqq\mu_{J}+\e\lambda_{J}\ge 0$, and $\supp \mu_J\subseteq \supp \hat{\mu}_J$.\footnote{Such a decomposition of \(\lambda\) exists whenever $\varepsilon$ is small enough: Since $\mu+\lambda$ is positive, $\supp(\mu) \subseteq \supp(\mu+\varepsilon\lambda)$ for $\varepsilon>0$ small enough. For each $J\in V$, let $\lambda_J\coloneqq\frac{1}{\varepsilon|V|}(\mu+\varepsilon\lambda-|V|\mu_J)$. Then $\hat{\mu}_J= \mu_J+\varepsilon \lambda_J=\frac{1}{|V|}(\mu+\varepsilon\lambda)\ge 0$ and $\supp(\mu_J)\subseteq \supp(\mu)\subseteq\supp(\hat{\mu}_J)$.} Since $r(\mu),r(\mu+\lambda)\in \aff\{x_i|i\in I\}$, there are $\gamma_i\in\R$ such that $\int x \,\mathrm d \lambda = \sum_{i\in I} \gamma_i x_i$ and $\sum_{i\in I} \gamma_i = \lambda(X)$.\footnote{Let $\alpha_i$ satisfy $r(\mu+\lambda)=\sum_i \alpha_i x_i$ and $\sum_i \alpha_i =1$ and let $\beta_i$ satisfy $r(\mu)=\sum_i \beta_i x_i$ and $\sum_i \beta_i=1$. Then $\gamma_i\coloneqq (\alpha_i-\beta_i)\mu(X)+\alpha_i \lambda(X)$ satisfy these properties.}

Consider a spanning tree $T$ of $G$.
For each leaf $J$ of $T$, define a weight $w_{J} = \sum_{i\in J}\gamma_i - \lambda_{J}(X)$
 and let
\[\mut_{J} = \muh_{J} + \pi_J ,  \]
where $\pi_{J}$ is a measure chosen such that $\pi_J(X)=\e w_J$, $\mut_{J}\ge 0$, $ \pi_J\le \frac{1}{|I|+1}\muh_{J'}$ (where $J'$ is the neighbor of $J$ in $T$), and 
\[ \e\int x \,\mathrm d \lambda_J + \int x \,\mathrm d\pi_J = \e\sum_{i\in J} \gamma_i x_i. \]
Such $\pi_{J}$ exist by Lemma \ref{lemma:shifting} whenever $\e$ is chosen small enough.

Proceeding inductively, consider the tree obtained by deleting all leaves from the previous tree. For each leaf $J$ of the new tree, denote by $L(J)$ the set of leaves of the previous tree that were connected to $J$, and note that for each leaf $J$ of the new tree, $\muh_J-\sum_{K\in L(J)}\pi_K\ge 0$ and $\supp(\muh_J)\subseteq \supp(\muh_J-\sum_{K\in L(J)}\pi_K)$. Define the weight $w_J = \sum_{K\in L(J)} w_K+ \sum_{i\in {J}}\gamma_i - \lambda_{J}(X)$ for each leaf $J$ of the new tree. As long as there are at least 3 vertices remaining, define
 \[\mut_J = \muh_J - \sum_{K\in L(J)}  \pi_K  + \pi_J \]
 where the $\pi_{J}$'s are chosen so that $\pi_J(X)=\e w_J$, $\mut_{J}\ge 0$, $ \pi_J \le \frac{1}{|I|+1} \muh_{J'}$ (where $J'$ is the neighbor of $J$ in the remaining tree), and 
\[\label{eq:bary}\tag{\(A1\)}
\e \int \,\mathrm x d \lambda_J -\sum_{K\in L(J)}  \int x \,\mathrm d\pi_K +  \int x \,\mathrm d\pi_J =\e \sum_{i\in J} \gamma_i x_i.\]
Such a choice exists by Lemma \ref{lemma:shifting} whenever $\e>0$ is small enough. We repeat this procedure until we obtain a tree with at most two vertices.

If only two nodes are left in the remaining tree (denote them by $\hat{J}$ and $\bar{J}$), define
\begin{align*}
\mut_{\hat{J}} &= \muh_{\hat{J}} - \sum_{K\in L(\hat{J})} \pi_K + \pi_{\hat{J}} \\
 \mut_{\bar{J}} &= \muh_{\bar{J}} - \sum_{K\in L(\bar{J})}   \pi_K - \pi_{\hat{J}}
\end{align*}
where $\pi_{\hat{J}}$ is chosen so that $\pi_{\hat{J}}(X)=\e w_{\hat{J}}$, $\mut_{\hat{J}}\ge 0$,  $\mut_{\bar{J}}\ge 0$, and
\[ \e \int x \,\mathrm d \lambda_{\hat{J}} -\sum_{K\in L({\hat{J}})}  \int x \,\mathrm d\pi_K + \int x \,\mathrm d\pi_{\hat{J}} = \e\sum_{i\in {\hat{J}}} \gamma_i x_i. \]
 
 If only one node, denoted by $\bar{J}$, is left in the remaining tree, define
 \[ \mut_{\bar{J}} = \mu_{\bar{J}}+\e \lambda_{\bar{J}} - \sum_{K \in L(\bar{J})}  \pi_K .  \]

In either case, it follows that $\sum_{J\in V} \mut_J=\sum_{J\in V} \muh_J$ by construction. Moreover, for all $J\neq \bar{J}$, the barycenter of $\mut_J$ is in $\aff\{x_i|i\in J\}$.\footnote{This holds because $r(\mu_J)\in \aff\{x_i|i\in J\}$, $\mut_J(X)=\mu_J(X)+ \e \sum_{i\in J} \gamma_i$ and $\int x \,\mathrm d\mut_J = \mu_J(X) r(\mu_J)+ \e \sum_{i\in J} \gamma_i x_i$.}

It remains to verify that the barycenter of $\mut_{\bar{J}}$ is in $\aff\{x_i|i\in \bar{J}\}$.
For $J\in V$, let $b_J = \int x \,\mathrm d\pi_J$.
For any node $J\neq \bar{J}$ we obtain from \eqref{eq:bary} that
\[\label{eq:b_J}\tag{\(A2\)}
      b_J = \e\sum_{i\in J} \gamma_i x_i - \e\int x \,\mathrm d \lambda_J + \sum_{K\in L(J)} b_K,\]
where $L(J)=\emptyset$ if $J$ is a leaf of $T$.
Therefore
\begin{align*}
\left[\mu_{\bar{J}}(X) + \e\lambda_{\bar{J}}(X) -\e\sum_{K\in L(\bar{J})} w_K - \e w_{\hat{J}} \right]r(\mut_{\bar{J}}) &= r(\mut_{\bar{J}}) \mut_{\bar{J}}(X)  \\ 
&= r(\mu_{\bar{J}})\mu_{\bar{J}}(X) + \e\lambda_{{\bar{J}}}(X) r(\lambda_{\bar{J}}) - \sum_{K\in L({\bar{J}})}  b_K\\
&= r(\mu_{\bar{J}})\mu_{\bar{J}}(X) + \e\int x \, \mathrm{d} \lambda - \e\sum_{i\not\in \bar{J}} \gamma_i x_i,
\end{align*}
where the first equality follows from the definition of $\pi_{K}$ and the third equality follows from using \eqref{eq:b_J} repeatedly.

Suppose $\hat{J}$ and $\bar{J}$ are the remaining vertices.
Since $\sum_{K\in L(\bar{J})} w_K + w_{\hat{J}} = \sum_{i\not\in \bar{J}} \gamma_i - \sum_{J\neq \bar{J}} \lambda_J(X)$,
and $ \int x \,\mathrm d \lambda = \sum_i \gamma_i x_i$, this implies
\begin{align*}
r(\mut_{\bar{J}}) = \frac{1}{\mu_{\bar{J}}(X)+ \e\sum_{i\in {\bar{J}}} \gamma_i}\left[ \mu_{\bar{J}}(X) r(\mu_{\bar{J}}) + \e\sum_{i\in{\bar{J}}} \gamma_i x_i \right].
\end{align*}
Since the barycenter of $\mu_{\bar{J}}$ lies in $\aff\{x_i|i\in {\bar{J}}\}$, we conclude that the barycenter of $\mut_{\bar{J}}$ lies in $\aff\{x_i|i\in {\bar{J}}\}$.

If $\bar{J}$ is the only remaining vertex, then $\sum_{K\in L(\bar{J})} w_K  = \sum_{i\not\in \bar{J}} \gamma_i - \sum_{J\neq \bar{J}} \lambda_J(X)$, which again implies
\begin{align*}
r(\mut_{\bar{J}}) = \frac{1}{\mu_{\bar{J}}(X)+ \e\sum_{i\in {\bar{J}}} \gamma_i}\left[ \mu_{\bar{J}}(X) r(\mu_{\bar{J}}) + \e\sum_{i\in{\bar{J}}} \gamma_i x_i \right].
\end{align*}
It follows that the barycenter of $\mut_{\bar{J}}$ lies in $\aff\{x_i|i\in {\bar{J}}\}$.
\end{proof}


\begin{proof}[Proof of Proposition \ref{prop:necessity_extreme}]
Let $\P $ be a partition of $X$ such that each $A\in\P $ is convex, satisfies $\mu(A)>0$, and  $\nu|_{A}\prec \mu|_{A}$. Moreover,
assume there is no finer partition satisfying these properties.\footnote{%
The trivial partition $\P =\{X\}$ satisfies all conditions. Moreover, there
is an upper bound on the number of partition elements in any such partition
since the support of $\nu$ is finite. Therefore, there is a partition that
satisfies all conditions and cannot be refined.}

Suppose that for some $A\in\P $, the support of $\nu|_A$ is not affinely
independent. Let $\{x_i|i\in I\}$ denote the support of $\nu|_A$. Also, let $%
\mu_i$ satisfy $r(\mu_i)=x_i$, $\mu_i(X)=\nu(\{x_i\})$, and $\mu|_A=\sum_i
\mu_i$ (such $\mu_i$ exist by Lemma \ref{lemma:cartier}).

Construct a graph $G_1$ with vertex set $I$ and an edge between $i$ and $j$ if 
\begin{equation*}
\interior(\ch(\supp(\mu_i)))\cap \interior(\ch(\supp(\mu_i)))\neq \emptyset.
\end{equation*}

Recursively, given a graph $G_n$ construct a new graph $G_{n+1}$ with vertex
set being the set of connected components in $G_n$ (for each connected
component in $G_n$ there is one vertex in $G_{n+1}$; each vertex is labeled by a subset of $I$ corresponding to the set of
vertices it represents). Add an edge between
vertices ${C}$ and ${C^{\prime }}$ in $G_{n+1}$ if 
\begin{equation*}
\interior(\ch(\bigcup_{i\in C}\supp(\mu_i)))\cap \interior(\ch(\bigcup_{i\in
C^{\prime }}\supp(\mu_i)))\neq \emptyset.
\end{equation*}
Stop this procedure of constructing graphs once it yields a graph 
$G_N$ with a vertex $C$ such that $\{x_i|i\in C\}$ is affinely dependent.

We first claim that the procedure always stops with such a graph $G_N$.
 Suppose the procedure does not stop. Up to some iteration $M$, each iteration reduces the number of vertices by at least one; after that,  the number of vertices stays constant, which implies that there are no edges. Let $G_M$ denote this graph without edges. We claim that $G_M$ has exactly one vertex.
Suppose not and let $J_1,...,J_m$ with $m\ge 2$ denote its vertices. For $k=1,...,m$, $A_k := \ch(\bigcup_{i\in J_k} \supp(\mu_i))$ is convex and satisfies $\mu(A_k)>0$. Since there are no edges in $G_M$, $\interior(A_j)\cap \interior(A_k)=\emptyset$ for $j,k=1,...,m$ with $j\neq k$. Since every hyperplane has $\mu$-measure 0, this implies $\nu|_{A_k}\prec \mu|_{A_k}$. 
This implies that the collection $\{A_1,...,A_m\}$ gives rise to a partition of $\P'$ of $X$ that is finer than $\P$, contradicting our initial hypothesis. We conclude that $G_M$ has one vertex; since $\{x_i|i\in I\}$ is affinely dependent, this contradicts the hypothesis that the procedure did not stop.

This implies that graph $G_{N-1}$ has a connected component corresponding to vertex $C$ in graph $G_N$. Let $T$ denote a spanning tree of this connected component. We can assume that $T$ is minimal in the sense that for any leaf $J$ of $T$, $\{x_i|i\in J\}$ and $\{x_i|i\in C\setminus J\}$ are affinely independent: Indeed, for any leaf $J$ of $T$, $\{x_i|i\in J\}$ is affinely independent since the procedure would have stopped with graph $G_{N-1}$ otherwise. If $\{x_i|i\in C\setminus J\}$ is affinely dependent, consider the graph $T'$ obtained by deleting vertex $J$ from graph $T$. For any leaf $J'$ of $T'$, either $\{x_i|i\in J'\}$ and $\{x_i|i\in C\setminus (J\cup J')\}$ are affinely independent, or we can reduce the tree $T'$ further until we obtain a tree with the desired properties.

Now let $\{J_1,...,J_m\}$ denote the vertices of $T$, choose a leaf $J$ of  $T$ (without loss of generality, assume $J=J_1$), and let $\mu_{C\setminus J}=\sum_{i\in C\setminus J}\mu_i$. 
Since $\{x_i|i\in C\}$ is affinely dependent, there is $\beta_i$ such that $\sum_{i\in C} \beta_i=0$ and $%
\sum_{i\in C} \beta_i x_i=0$. This implies 
\begin{equation*}
d:= \frac{\sum_{i\in J} \beta_i x_i}{\sum_{i\in J} \beta_i} = \frac{%
\sum_{i\in C\setminus J} \beta_i x_i}{\sum_{i\in C\setminus J} \beta_i}.
\end{equation*}
Lemma \ref{lemma:shifting} implies that for any $\varepsilon>0$ small enough,
there is a measure $\pi$ with $\mu_J+\pi\ge 0 $, $ \mu_{C\setminus J}-\pi\ge 0$, $%
\pi(X)=\varepsilon$ and $\int x \,\mathrm d\pi=\varepsilon d$. It follows that 
\begin{align*}
r(\mu_J+\pi) &= \frac{1}{\mu_J(X)+\varepsilon}\left[\int x \,\mathrm d\mu_J + \int x \,\mathrm d\pi\right]=
\frac{1}{\mu_J(X)+\varepsilon}\Big[ \mu_J(X) r(\mu_J) + \varepsilon d\Big] \\
r(\mu_{C\setminus J}-\pi) &= \frac{1}{\mu_{C\setminus J}(X)-\varepsilon}\left[\int x
\,\mathrm d\mu_{C\setminus J} - \int x \,\mathrm d\pi\right]= \frac{1}{\mu_{C\setminus J}(X)-\varepsilon}\Big[ \mu_{C\setminus J}(X) r(\mu_{C\setminus J}) - \varepsilon d\Big].
\end{align*}

Therefore, $r(\mu_J+\pi)\in \aff\{x_i|i\in J\}$ and $r(\mu_{C\setminus J}-\pi)\in\aff\{x_i|i\in C\setminus J\}$.
Lemma \ref{lemma:decompo} then implies that whenever $\varepsilon$ is small enough, we can decompose $\mu_{C\setminus J}-\pi$ into positive measures $\{\mut_{J_k}\}_{k=2}^m$ with $\sum_{k=2}^m \mut_{J_k}=\mu_{C\setminus J}-\pi$ and $r(\mut_{J_k})\in\aff\{x_i|i\in J_k\}$. We can decompose $\mu_J+\pi$ in a similar manner. Applying these decompositions repeatedly to graphs $G_{N-2}$, $G_{N-3}$, etc., this yields measures $\{\tilde{\mu}_i\}_{i\in C}$ with $\sum_{i\in C} \tilde{\mu}_i = \sum_{i\in C} \mu_i$ and $r(\tilde{\mu}_i)=x_i$ that satisfy $\tilde{\mu}_i(X)\neq \mu_i(X)$ for some $i\in C$ (since $\sum_{i\in J}\mut_i(X)-\mu_i(X) = \pi(X)>0$). 
We can then define $\tilde{\nu}$ to be a positive measure with the same support as $\nu|_A$ that satisfies $\tilde{\nu}(\{x_i\})=\tilde{\mu}_i(X)$ for $i\in C$ and $\tilde{\nu}(\{x_i\})=\nu(\{x_i\})$ for $i\in I\setminus C$. This is a fusion of $\mu|_A$ by Lemma \ref{lemma:cartier} and it satisfies $\tilde{\nu}\neq \nu|_A$.

Repeating the procedure with $-d$ in place of $d$ yields another fusion $\hat{\nu}\neq \nu|_A$. We claim that $\nu|_A= 1/2 (\tilde{\nu}+\hat{\nu})$:
Note that 
\begin{align}
&\sum_{i\in J} \nu(\{x_i\}) x_i =  \int x \,\mathrm d\mu_J = 1/2 \Big[ \int x \,\mathrm d\mu_J+\varepsilon d +\int x \,\mathrm d\mu_J-\varepsilon d \Big] \nonumber \\
= &1/2 \sum_{i\in J} [\tilde{\nu}(\{x_i\})+\hat{\nu}(\{x_i\})] x_i, \label{eq:convex_comb_fusion}\tag{\(A3\)}
\end{align}
where the first equality follows since $\sum_{i\in J} \nu(\{x_i\})$ is a fusion of $\mu_J$ and the final equality follows since, by construction of $\tilde{\nu}$ and $\hat{\nu}$, $\int x \,\mathrm d[\sum_{i\in J} \tilde{\nu}(\{x_i\})]=\int x \,\mathrm d\mu_J+\varepsilon d$ and $\int x \,\mathrm d[\sum_{i\in J} \hat{\nu}(\{x_i\})]=\int x \,\mathrm d\mu_J-\varepsilon d$.
Since $\{x_i|i\in J\}$ is affinely independent, \eqref{eq:convex_comb_fusion} implies that for each $i\in J$, $\nu(\{x_i\})= 1/2 [\tilde{\nu}(\{x_i\})+\hat{\nu}(\{x_i\})]$. Analogous arguments apply for $i\in C\setminus J$, which establishes the claim.

This implies that $\nu$ is not an extreme point of $F_{\mu}$, contradicting our initial hypothesis. We conclude that $\supp(\nu|_A)$ is affinely independent for each $A\in \P$.
\end{proof}


\subsection{Proof of Proposition \protect\ref{prop:convexpartfiner}}


\begin{proof}
Let \(G\) denote the set of probability measures that are fusions of \(\mu\) and have at most \(K\) points in their support.
We show first using Zorn's lemma that there is a maximal measure in \(G\) according to the convex order: The convex order is a partial order. Let $\{g_i\}_{i\in I}$ be a totally ordered subset of \(G\), which can be viewed as a net by setting $i<j$ if $g_i\prec g_j$. Since \(G\) is compact (e.g., in the weak$^*$-topology), there is subnet with limit $g_{\infty}$. Then $g_i\prec g_{\infty}$ because for any convex continuous function $f$ and $i<j$, $\int f \mathrm dg_i\le \int f \mathrm dg_j$. Therefore, every totally ordered subset of \(G\) has an upper bound in \(G\). It follows from Zorn's lemma that there is a maximal element in \(G\).

Let \(\lambda\) denote the maximal element of \(G\), let $\{x_1,...,x_K\}$ denote its support points, and suppose \(\lambda\) is not convex partitional. Using Lemma \ref{lemma:cartier} 
there is a decomposition of $\mu$ into positive measures $\{\lambda_i\}$ such that $\mu=\sum_{i=1}^K \lambda_i$,  $r(\lambda_i)=x_i$, and $\lambda_i(X) = \lambda(\{x_i\})$. Since $\lambda$ is not convex partitional, for some $i,j$, $\interior(\ch(\supp (\lambda_i)))\cap \interior(\ch(\supp( \lambda_j))) \neq \emptyset$. It follows from Lemma \ref{lemma:shifting} 
that there exist a measure \(\alpha\) and $\varepsilon>0$ such that $\alpha(X)=0$, $-\lambda_j\le \alpha\le \lambda_i$, and $\int x \mathrm d \alpha_i = \varepsilon (x_j-x_i)$. If we fuse for each $k\neq i,j$ the measure $\lambda_k$ to its barycenter and fuse $\lambda_i- \alpha_i\ge 0$ and $\lambda_j + \alpha_i/2\ge 0$ to their barycenters (if two of these measures have the same barycenter, decrease the value of $\varepsilon$ slightly), we obtain an element of \(G\) that is larger than \(\lambda\) in the convex order, a contradiction.
\end{proof}

\begin{proof}[Proof of Proposition \protect\ref{lem:feasflows}]
\ref{flowi} Assume \(\nu\) is not an extreme point. Then there are $\nu_1,\nu_2\in F_{\mu}$
such that $\nu_1\neq \nu_2$ and $\nu=\frac{1}{2} \nu_1+\frac{1}{2} \nu_2$. Let $\{x_1,...,x_m\}$ denote the support of \(\nu\) and let $\mu^i_j$ be positive measures such that, for $i=1,2$, $\mu=\sum_j \mu^i_j$, $\nu_i(\{x_j\})=\mu^i_j(X)$, and $r(\mu^i_j)=x_j$ (such measures exist by Lemma \ref{lemma:cartier}). Define $\mu^i_P\coloneqq\sum_{j:x_j\in P} \mu^i_j$ and $u_P\coloneqq\frac{1}{2}\mu^1_P +\frac{1}{2}\mu^2_P - \mu|_P$.
It follows that $\sum_P u_P=0$. Moreover, $u_P$ is positive on $X\setminus P$ and, since $\mu^i_P\le \mu$, it satisfies \eqref{eq:u_pos}. Analogous arguments establish \eqref{eq:u_neg}. Finally, $u_P(X)=0$ and $\int x du_P=0$. Since $\nu_1\neq \nu\neq \nu_2$, $u_P\neq 0$ for some $P$.

\ref{flowii} Conversely, suppose there is a solution $\{u_P\}$ to \eqref{eq:u_pos}--\eqref{eq:sum_u} satisfying $u_Q\neq 0$ for some $Q\in\P$. Enumerate the elements of \(\P\) as $\P=\{P_1,...,P_m\}$. Consider a fusion \(\nu\) of \(\mu\) that satisfies $\nu_P\prec \mu|_P$, the support of $\nu|_P$ is affinely independent and spans \(X\), all points in the support of $\nu|_P$ are contained in a $\delta$-ball around the barycenter of $\nu|_P$.

Choose $\varepsilon\in (0,1)$ small enough and define, for $i=1,...,m$,  
\begin{align*}
	\mu^i_{P_i}\coloneqq\mu|_{P_i}+\varepsilon u_{P_i}|_{X\setminus {P_i}}\\
	\mu^i_{P}\coloneqq\mu|_P-\varepsilon u_{P_i}|_{P} \text{ for } P\neq P_i
\end{align*}

By construction, $\mu^i_P$ is a positive measure for all $P\in\P$ and $\sum_{P\in\P} \mu^i_P = \mu$. If we define $\nu^i_P$ to have support $\supp(\nu)\cap P$ and the same barycenter and mass as $\mu^i_P$, then $\nu^i_P$ is a fusion of $\mu^i_P$ whenever $\varepsilon>0$ is chosen small enough (if it wasn't a fusion, we could choose the $\delta$-ball above smaller and redo the exercise). Consequently, $\nu^i\coloneqq\sum_{P\in\P} \nu^i_P$ is a fusion of \(\mu\) and $\nu^i\neq \nu$ for some $i$.




We claim that $\frac{1}{m} \sum_{i=1}^m \nu^i =\nu$, which implies that \(\nu\) is not an extreme point. To establish the claim, we argue that $\frac{1}{m} \sum_{i=1}^m \nu^i(P)=\nu(P)$ and $\frac{1}{m} \sum_{i=1}^m \int_P x d\nu^i=\int_P x d\nu$ for all $P\in\P$. Since the support of $\nu^i$ and \(\nu\) coincides on each $P$ and is affinely independent, this implies that $\frac{1}{m} \sum_{i=1}^m \nu^i=\nu$. That is,
 
 \begin{align*}
 \frac{1}{m} \sum_{i=1}^m \nu^i(P_j) &= \mu|_{P_j} + \frac{1}{m} \varepsilon u_{P_j}|_{X\setminus {P_j}}(X) - \frac{1}{m} \varepsilon \sum_{i:i\neq j} u_{P_i}|_{P_j}(X)\\
 &= \mu|_{P_j} + \frac{1}{m} \varepsilon u_{P_j}|_{X\setminus {P_j}}(X) + \frac{1}{m} \varepsilon u_{P_j}|_{P_j}(X) \text{ (since }\sum_{i} u_{P_i}|_{P_j}=0 \text{ )} \\
 &= \mu|_{P_j} + \frac{1}{m} \varepsilon u_{P_j}(X)  \\
 &= \nu(P_j)
 \end{align*}

 Also, 
 \begin{align*}
 \frac{1}{m} \sum_{i=1}^m \int_{P_j} x d\nu^i &= \frac{1}{m} \sum_{i=1}^m \int_{P_j} x d\mu^i \\
 &= \int_{P_j} x d\mu  + \varepsilon/m \int_{X\setminus P_j} x d u_{P_i} - \varepsilon/m \sum_{i:i\neq j} \int_{P_j} x d u_{P_i} \\
&= \int_{P_j} x d\mu  + \varepsilon/m \int_{X\setminus P_j} x d u_{P_i} + \varepsilon/m  \int_{P_j} x d u_{P_j} \\
&= \int_{P_j} x d\nu.  
 \end{align*}

\end{proof}

\begin{proof}[Proof of \autoref{cor:suff_for_exposed}]
    Let $\P_1\subseteq \P$ be a maximal subset such that for all $A\in\P_1$, $\mu|_A=\nu|_A$ and let $\P_2:= \P\setminus \P_1$.

    Since \(\P\) is a power diagram, there is a convex function $p\colon X\rightarrow \R$ such that for each the restriction of \(p\) to $A$ is affine and if \(p\) is affine on $B\subseteq X$ then there is $A\in \P$ with $B\subseteq A$. We first adjust $p$ to obtain another convex function $q$ as follows. For any $A\in\P_1$, choose a Lipschitz-continuous function $c_A$ that is strictly convex on the interior of $A$ and equals 0 outside the interior of $A$ (see \autoref{l:existence_strictly_convex}). Then define
    \[q\left(x\right) := p\left(x\right) + \sum_{A\in\P_1} k_A c_A(x)\]
    where the constants $k_A$ are chosen such that $q$ is convex. 
    
    Now we define a Lipschitz-continuous function \(u\colon X\rightarrow \R\) by
\[u\left(x\right) := q\left(x\right) - \inf_{y\in \supp(\nu)} \norm{x-y}\text{.}\]
Note that by definition $u\left(x\right)=q\left(x\right)$ if $x\in \supp(\nu )$ and that $u\left(x\right)<p\left(x\right)$ for $x\not\in  \supp (\nu )$.

We claim that \(\nu\) is the unique solution to
\[\max_{\lambda\in F_{\mu}} \int u(x) \,\mathrm d \lambda(x)\tag{L} \label{eq:util_maxapp2}\]

\noindent \textbf{Step 1:} \textit{\(\nu\) is a solution to problem \eqref{eq:util_maxapp2}.}

By construction, $\int u(x) \,\mathrm d \nu(x) = \int q\left(x\right) \,\mathrm d\nu(x) $. Moreover, $\int q\left(x\right)\,\mathrm d\nu = \int q\left(x\right) \,\mathrm d\mu$ since for each $A\in\P$, either $\mu|_A=\nu|_A$ or $\nu|_{A}\in F_{\mu|_{A}}$ and \(p\) is affine on $A$. Therefore, 
\[
\int u(x) \,\mathrm d \nu(x) = \int q\left(x\right) \,\mathrm d\mu(x).\label{eq:unu_pmu2}\tag{\(A4\)}\]

On the other hand, since $q$ is convex and $q\ge u$, for any \(\lambda\in F_{\mu}\) we obtain
\[ \int u(x) \,\mathrm d \lambda (x) \le \int q\left(x\right) \,\mathrm d \lambda(x) \le \int q\left(x\right) \,\mathrm d\mu(x)\text{,}\]
where the second inequality follows from Jensen's inequality. We conclude that \(\nu\) is a solution to \eqref{eq:util_maxapp2}.

\medskip

\noindent \textbf{Step 2:} \textit{There is no other solution to problem \eqref{eq:util_maxapp2}.}

Let \(\lambda\in F_{\mu}\) solve problem \eqref{eq:util_maxapp2}.

First, we establish that $\lambda|_P\preceq \mu|_P$ for all \(P \in \P\): Suppose not and let $D_x$ denote the dilations that carry $\lambda$ to $\mu$. Then there is a set of positive $\nu-$measure $B$ such that the dilations $D_x$ for all $x\in B$ have the property that $\supp D_x$ is not contained in an element of $\P$. Therefore, $q$ is not affine on $\supp D_x$, and we obtain $\int q(x) \,\mathrm d\lambda <\int \int q(y) \,\mathrm dD_x \,\mathrm d\lambda(x)=\int q(x) \,\mathrm d\mu(x)$. Therefore, $\lambda$ is not a solution, a contradiction.

The same argument establishes $\lambda|_A=\mu|_A$ for all $A\in \P_1$ since $q$ is strictly convex on the interior of $A$.

Next, we claim that $\supp \lambda \subseteq \supp \nu$.
 If not, there is $A\in\P_2$ and $x\in A\cap \supp\nu$ with $x\not\in \supp \lambda$. Then $\int u(x) \,\mathrm d \lambda(x)< \int q\left(x\right) \,\mathrm d \lambda(x)$ since $u(x)<q\left(x\right)$ for all $x\not\in\supp \nu$. This implies
\[\label{eq:inequal}\tag{\(A5\)}
  \int u(x) \,\mathrm d \lambda(x)<\int q\left(x\right) \,\mathrm d \lambda(x)\le \int q\left(x\right) \,\mathrm d \mu(x)=\int u(x) \,\mathrm d\nu(x)\text{,}\]
where the second inequality follows since \(p\) is convex and \(\lambda\in F_{\mu}\), and the equality follows from \eqref{eq:unu_pmu2}. We conclude that $\lambda$ does not solve problem \eqref{eq:util_maxapp2}, a contradiction.

Finally, since the support of $\nu$ is affinely independent on each $A\in \P_2$, this yields $\nu=\lambda$.
\end{proof}

\section{Sufficient conditions for extremal fusions}\label{sec:appendix_B}

We next explore sufficient conditions for a finitely supported measure $\nu $
to be an extreme point of $F_{\mu }$. We use these conditions to verify that
a particular fusion that is not Lipschitz-exposed is an exposed point.

\begin{definition}
Given a positive measure $\mu $ and a partition $\mathcal{P}$ of $X$ into
finitely many convex subsets, we say that a collection of measures $%
\{u_{P}\}_{P\in \mathcal{P}}$ is a \emph{feasible flow for }$\mathcal{P}$ if
it satisfies 
\begin{align}
0& \leq u_{P}|_{X\setminus P}\leq \mu |_{X\setminus P}  \label{eq:u_pos} \\
-\mu |_{P}& \leq u_{P}|_{P}\leq 0  \label{eq:u_neg} \\
u_{P}(X)& =0 \\
\int_{X}xdu_{P}& =0 \\
\sum_{P}u_{P}& =0  \label{eq:sum_u}
\end{align}
\end{definition}

\begin{proposition}
\label{lem:feasflows} Fix an absolutely continuous measure $\mu $ 
and let $\mathcal{P}$ be a partition of $X$ into finitely many convex sets.

\begin{enumerate}
\item \label{flowi} Suppose that $\nu \in F_{\mu }$ satisfies $\nu
|_{P}\preceq \mu _{P}$ for all $P\in \mathcal{P}$ and that the support of $%
\nu |_{P}$ is affinely independent. If $u_{P}\equiv 0$ for all $P\in 
\mathcal{P}$ is the unique feasible flow for $\mathcal{P}$ then $\nu $ is an
extreme point of $F_{\mu }$.

\item \label{flowii} Suppose there exists a non-zero feasible flow for $%
\mathcal{P}$. Then there exists a fusion $\nu \in F_{\mu }$ that is not an
extreme point of $F_{\mu }$ and that satisfies

\begin{enumerate}
\item For each $P\in \P $, $\nu |_{P}\preceq \mu |_{P}$; and

\item The support of $\nu |_{P}$ is affinely independent.
\end{enumerate}
\end{enumerate}
\end{proposition}

\begin{proof}[Proof of Proposition \ref{lem:feasflows}]
Part \ref{flowi}: Assume \(\nu\) is not an extreme point. Then there are $\nu_1,\nu_2\in F_{\mu}$
such that $\nu_1\neq \nu_2$ and $\nu=\frac{1}{2} \nu_1+\frac{1}{2} \nu_2$. Let $\{x_1,...,x_m\}$ denote the support of \(\nu\) and let $\mu^i_j$ be positive measures such that, for $i=1,2$, $\mu=\sum_j \mu^i_j$, $\nu_i(\{x_j\})=\mu^i_j(X)$, and $r(\mu^i_j)=x_j$ (such measures exist by Lemma \ref{lemma:cartier}). Define $\mu^i_P\coloneqq\sum_{j:x_j\in P} \mu^i_j$ and $u_P\coloneqq\frac{1}{2}\mu^1_P +\frac{1}{2}\mu^2_P - \mu|_P$.
It follows that $\sum_P u_P=0$. Moreover, $u_P$ is positive on $X\setminus P$ and, since $\mu^i_P\le \mu$, it satisfies \eqref{eq:u_pos}. Analogous arguments establish \eqref{eq:u_neg}. Finally, $u_P(X)=0$ and $\int x \,\mathrm du_P=0$. Since $\nu_1\neq \nu\neq \nu_2$, $u_P\neq 0$ for some $P$. Hence, there is a non-zero feasible flow.

Part \ref{flowii}: Conversely, suppose there is a solution $\{u_P\}$ to \eqref{eq:u_pos}--\eqref{eq:sum_u} satisfying $u_Q\neq 0$ for some $Q\in\P$. Enumerate the elements of \(\P\) as $\P=\{P_1,...,P_m\}$. Consider a fusion \(\nu\) of \(\mu\) that satisfies $\nu_P\prec \mu|_P$, the support of $\nu|_P$ is affinely independent and spans \(X\), all points in the support of $\nu|_P$ are contained in a $\delta$-ball around the barycenter of $\nu|_P$.

Choose $\varepsilon\in (0,1)$ small enough and define, for $i=1,...,m$,  
\begin{align*}
    \mu^i_{P_i}\coloneqq\mu|_{P_i}+\varepsilon u_{P_i}|_{X\setminus {P_i}}\\
    \mu^i_{P}\coloneqq\mu|_P-\varepsilon u_{P_i}|_{P} \text{ for } P\neq P_i
\end{align*}

By construction, $\mu^i_P$ is a positive measure for all $P\in\P$ and $\sum_{P\in\P} \mu^i_P = \mu$. If we define $\nu^i_P$ to have support $\supp(\nu)\cap P$ and the same barycenter and mass as $\mu^i_P$, then $\nu^i_P$ is a fusion of $\mu^i_P$ whenever $\varepsilon>0$ is chosen small enough (if it wasn't a fusion, we could choose the $\delta$-ball above smaller and redo the exercise). Consequently, $\nu^i\coloneqq\sum_{P\in\P} \nu^i_P$ is a fusion of \(\mu\) and $\nu^i\neq \nu$ for some $i$.




We claim that $\frac{1}{m} \sum_{i=1}^m \nu^i =\nu$, which implies that \(\nu\) is not an extreme point. To establish the claim, we argue that $\frac{1}{m} \sum_{i=1}^m \nu^i(P)=\nu(P)$ and $\frac{1}{m} \sum_{i=1}^m \int_P x d\nu^i=\int_P x d\nu$ for all $P\in\P$. Since the support of $\nu^i$ and \(\nu\) coincides on each $P$ and is affinely independent, this implies that $\frac{1}{m} \sum_{i=1}^m \nu^i=\nu$. That is,
 
 \begin{align*}
 \frac{1}{m} \sum_{i=1}^m \nu^i(P_j) &= \mu|_{P_j} + \frac{1}{m} \varepsilon u_{P_j}|_{X\setminus {P_j}}(X) - \frac{1}{m} \varepsilon \sum_{i:i\neq j} u_{P_i}|_{P_j}(X)\\
 &= \mu|_{P_j} + \frac{1}{m} \varepsilon u_{P_j}|_{X\setminus {P_j}}(X) + \frac{1}{m} \varepsilon u_{P_j}|_{P_j}(X) \text{ (since }\sum_{i} u_{P_i}|_{P_j}=0 \text{ )} \\
 &= \mu|_{P_j} + \frac{1}{m} \varepsilon u_{P_j}(X)  \\
 &= \nu(P_j)
 \end{align*}

 Also, 
 \begin{align*}
 \frac{1}{m} \sum_{i=1}^m \int_{P_j} x d\nu^i &= \frac{1}{m} \sum_{i=1}^m \int_{P_j} x d\mu^i \\
 &= \int_{P_j} x d\mu  + \varepsilon/m \int_{X\setminus P_j} x d u_{P_i} - \varepsilon/m \sum_{i:i\neq j} \int_{P_j} x d u_{P_i} \\
&= \int_{P_j} x d\mu  + \varepsilon/m \int_{X\setminus P_j} x d u_{P_i} + \varepsilon/m  \int_{P_j} x d u_{P_j} \\
&= \int_{P_j} x d\nu.  
 \end{align*}

\end{proof}



\begin{lemma}
\label{lemma:halfspaces_extreme_points} Suppose that $P\in \mathcal{P}$ is
the intersection of a half-space and $X$, i.e., $P=H\cap X$ for some
half-space of $\mathbb{R}^{n}$. If $\{u_{P}\}$ is a feasible flow then $%
u_{P}|_{P^{\prime }}=u_{P^{\prime }}|_{P}=0$ for all $P^{\prime }$.
\end{lemma}

If $u$ is a feasible flow then $\int x\,\mathrm{d}u_P=0$ and therefore $\int
x\,\mathrm{d}[-u_P|_P]=\sum_{P^{\prime }\neq P}\int x\,\mathrm{d}%
[u_P|_P^{\prime }]$. The result follows since if $u_{P}|_{P^{\prime }}\neq 0$
then the barycenter of $-u_{P}|_{P}$ lies in the interior of $P$, hence in
the interior of the half-space $H$, whereas the barycenter of $%
\sum_{P^{\prime }\neq P}u_{P}|_{P^{\prime }}$ lies in the complement of this
half-space. This violates $\int x\,\mathrm{d}u_P=0$, therefore $%
u_{P}|_{P^{\prime }}=0$.

Proposition \ref{lem:feasflows} and Lemma \ref%
{lemma:halfspaces_extreme_points} imply that the fusion we construct in our
second example is an extreme point.

\begin{lemma}
\label{lemma:approx_partition} Let $\mu $ be an absolutely continuous
measure with full support on $X$. 

Let $\mathcal{P}$ be a partition of $X$ into finitely many convex sets and
suppose there is a non-zero feasible flow $\{u_{P}\}_{P\in \mathcal{P}}$ for 
$\mathcal{P}$. Let $m_{PQ}\coloneqq u_{P}|_{Q}(X)$ and $x_{PQ}\coloneqq\int
xdu_{P}|_{Q}$. For all $P,Q\in \mathcal{P}$ with $m_{PQ}>0$, let $b_{PQ}%
\coloneqq\frac{x_{PQ}}{m_{PQ}}$.

If $\ \mathcal{P}^{\prime }$ is a partition of $X$ that approximates $%
\mathcal{P}$ in the sense that for each $P\in \mathcal{P}$ there is $%
P^{\prime }\in \mathcal{P}^{\prime }$ such that $P^{\prime }\subseteq P$
and, for all $Q\in \mathcal{P}$, $b_{QP}\in \interior P^{\prime }$ then
there is a non-zero feasible flow for $\P ^{\prime }$.
\end{lemma}

This result implies that if $\nu\in F_{\mu}$ is obtained from a partition $%
\P $ by collapsing the mass on each cell to its barycenter and if $\nu$ is
not an extreme point, then any fusion $\nu^{\prime }$ obtained from a
closeby partition $\P ^{\prime }$ is also not an extreme point. 
\begin{proof}
Let $k$ denote the number of partition elements in \(\P\) and fix $\delta>0$ small enough. For all $P\in \P$ and all $Q\in\P\setminus\{P\}$, let $P'$ and $Q'$ be the corresponding partition elements in $\P'$ that approximate $P$ and $Q$, respectively. Let
 $\tilde{u}_{Q'P'}$ be a non-negative measure with $\tilde{u}_{Q'P'}\le \frac{1}{k} \mu|_{P'}$, $\tilde{u}_{Q'P'}(X) = \delta m_{QP}$, and barycenter $\frac{1}{\delta m_{QP}}\int x d\tilde{u}_{Q'P'} = b_{QP}$ (if $m_{QP}>0$). Such measures exist by Lemma \ref{lemma:new_measure} whenever $\delta>0$ is chosen small enough. For all $P',Q'\in\P$ for which $\tilde{u}_{P'Q'}$ has not been defined this way, let $\tilde{u}_{P'Q'}$ be the zero measure.

For any $P'\in\P'$, define $\tilde{u}_{P'P'}= -\sum_{Q'\in\P'\setminus\{P'\}} \tilde{u}_{Q'P'}$ and define $\tilde{u}_{P'}= \sum_{Q'\in \P'} \tilde{u}_{P'Q'}$.

Then $\sum_{P'\in\P'}\tilde{u}_{P'}=\0$ by construction. Moreover, since $\sum_{Q\in\P} u_{Q}|_P$ is the zero measure, we obtain for any $P\in\P$ and its approximating partition cell $P'\in\P'$,
\[\delta u_P|_P(X)= -\delta\sum_{Q\in\P\setminus\{P\}} u_Q|_P(X) = - \sum_{Q'\in\P'\setminus\{P'\}} \tilde{u}_{Q'P'}(X) =\tilde{u}_{P'P'}(X).\]
This yields
\begin{align*}
\tilde{u}_{P'}(X) &= \sum_{Q'\in\P'} \tilde{u}_{P'Q'}(X) = \sum_{Q\in\P} \delta{u}_{P}|_Q(X) =\delta u_P(X) = 0.
\end{align*}
Analogous arguments establish that $\int x d\tilde{u}_{P'}=0$.
Since $0\le \tilde{u}_{P'}|_{X\setminus P'}\le \mu|_{X\setminus P'}$ and $-\mu|_{P'} \le \tilde{u}_{P'}|_{P'}\le 0$, it follows that $\{\tilde{u}\}$ is a non-zero feasible flow for $\P'$.
\end{proof}

\bibliography{sample}

\begin{thebibliography}{22}
\providecommand{\natexlab}[1]{#1}
\providecommand{\url}[1]{\texttt{#1}}
\expandafter\ifx\csname urlstyle\endcsname\relax
  \providecommand{\doi}[1]{doi: #1}\else
  \providecommand{\doi}{doi: \begingroup \urlstyle{rm}\Url}\fi

\bibitem[Arieli et~al.(2023)Arieli, Babichenko, Smorodinsky, and Yamashita]{arieli2023optimal}
Itai Arieli, Yakov Babichenko, Rann Smorodinsky, and Takuro Yamashita.
\newblock Optimal persuasion via bi-pooling.
\newblock \emph{Theoretical Economics}, 18\penalty0 (1):\penalty0 15--36, 2023.

\bibitem[Aurenhammer(1987)]{aurenhammer1987power}
Franz Aurenhammer.
\newblock Power diagrams: properties, algorithms and applications.
\newblock \emph{SIAM Journal on Computing}, 16\penalty0 (1):\penalty0 78--96, 1987.

\bibitem[Blackwell(1953)]{blackwell1953equivalent}
David Blackwell.
\newblock Equivalent comparisons of experiments.
\newblock \emph{The Annals of Mathematical Statistics}, 24\penalty0 (2):\penalty0 265--272, 1953.

\bibitem[Bunt(1934)]{bunt1934}
L.~N.~H. Bunt.
\newblock \emph{Bijdrage tot de theorie der convexe puntverzamelingen}.
\newblock Phd thesis, University of Groningen, 1934.

\bibitem[Cartier et~al.(1964)Cartier, Fell, and Meyer]{cartier1964comparaison}
Pierre Cartier, JMG Fell, and Paul-Andr{\'e} Meyer.
\newblock Comparaison, des mesures port{\'e}es par un ensemble convexe compact.
\newblock \emph{Bulletin de la Soci{\'e}t{\'e} Math{\'e}matique de France}, 92:\penalty0 435--445, 1964.

\bibitem[Dworczak and Kolotilin(2024)]{dworczak2019persuasion}
Piotr Dworczak and Anton Kolotilin.
\newblock The persuasion duality.
\newblock \emph{Theoretical Economics}, 19\penalty0 (4):\penalty0 1701--1755, 2024.

\bibitem[Dworczak and Martini(2019)]{dworczak2019simple}
Piotr Dworczak and Georgy Martini.
\newblock The simple economics of optimal persuasion.
\newblock \emph{Journal of Political Economy}, 127\penalty0 (5):\penalty0 1993--2048, 2019.

\bibitem[Eppstein(2014)]{eppstein2014}
David Eppstein.
\newblock Power diagram.
\newblock \url{https://commons.wikimedia.org/w/index.php?curid=20112050}, 2014.
\newblock CC0 License, Own work.

\bibitem[Fenchel(1929)]{fenchel1929}
W.~Fenchel.
\newblock {Ü}ber {K}rümmung und {W}indung geschlossener {R}aumkurven.
\newblock \emph{Mathematische Annalen}, 101:\penalty0 238--252, 1929.

\bibitem[Galichon(2018)]{galichon2018optimal}
Alfred Galichon.
\newblock \emph{Optimal Transport Methods in Economics}.
\newblock Princeton University Press, 2018.

\bibitem[G\"{a}rdenfors(2004)]{gardenfors2004conceptual}
Peter G\"{a}rdenfors.
\newblock \emph{Conceptual Spaces: The Geometry of Thought}.
\newblock A Bradford book. MIT Press, 2004.

\bibitem[Gutierrez(2016)]{gutierrez}
Cristian~E. Gutierrez.
\newblock \emph{The Monge-Ampere Equation}.
\newblock Progress in Nonlinear Differential Equations and their Applications. Birkh\"auser, 2016.

\bibitem[Hartenstine(2006)]{hartenstine}
David Hartenstine.
\newblock The {D}irichlet problem for the {M}onge-{A}mpere equation in convex (but not strictly convex) domains.
\newblock \emph{Electronic Journal of Differential Equations}, 2006\penalty0 (138):\penalty0 1--9, 2006.

\bibitem[Ivanov(2021)]{ivanov2021optimal}
Maxim Ivanov.
\newblock Optimal monotone signals in {B}ayesian persuasion mechanisms.
\newblock \emph{Economic Theory}, 72\penalty0 (3):\penalty0 955--1000, 2021.

\bibitem[Kamenica and Gentzkow(2011)]{kamenicagentzkow2011}
Emir Kamenica and Matthew Gentzkow.
\newblock Bayesian persuasion.
\newblock \emph{The American Economic Review}, 101\penalty0 (6):\penalty0 2590--2615, 2011.

\bibitem[Klee~Jr(1958)]{klee1958extremal}
Victor~L Klee~Jr.
\newblock Extremal structure of convex sets. ii.
\newblock \emph{Mathematische Zeitschrift}, 69\penalty0 (1):\penalty0 90--104, 1958.

\bibitem[Kleiner et~al.(2021)Kleiner, Moldovanu, and Strack]{kleiner2021extreme}
Andreas Kleiner, Benny Moldovanu, and Philipp Strack.
\newblock Extreme points and majorization: Economic applications.
\newblock \emph{Econometrica}, 89\penalty0 (4):\penalty0 1557--1593, 2021.

\bibitem[Lee and Santos(2017)]{lee_santos_subdivisions_2017}
Carl~W. Lee and Francisco Santos.
\newblock Subdivisions and triangulations of polytopes.
\newblock In Jacob~E. Goodman, Joseph O'Rourke, and Csaba Tóth, editors, \emph{Handbook of Discrete and Computational Geometry}, pages 425--448. CRC Press, Boca Raton, 3rd edition, 2017.

\bibitem[Rockafellar(1970)]{rockafellar1970}
R.~Tyrrell Rockafellar.
\newblock \emph{Convex Analysis}.
\newblock Princeton University Press, 1970.

\bibitem[Sertel(1989)]{sertel1989continuity}
Murat Sertel.
\newblock On the continuity of closed convex hull.
\newblock \emph{Mathematical Social Sciences}, 18\penalty0 (3):\penalty0 297--299, 1989.

\bibitem[Strassen(1965)]{strassen1965existence}
Volker Strassen.
\newblock The existence of probability measures with given marginals.
\newblock \emph{Annals of Mathematical Statistics}, 36\penalty0 (2):\penalty0 423--439, 1965.

\bibitem[Winkler(1988)]{winkler1988}
G.~Winkler.
\newblock Extreme points of moment sets.
\newblock \emph{Mathematics of Operations Research}, 13\penalty0 (4):\penalty0 581--587, 1988.

\end{thebibliography}

\end{document}